\documentclass[10pt,journal]{IEEEtran}
\ifCLASSINFOpdf
\else
\fi
\usepackage{epsfig}
\usepackage{amssymb}
\usepackage{amsmath}
\usepackage{color}
\usepackage{amsfonts}
\newtheorem{theorem}{Theorem}
\newtheorem{proposition}{Proposition}
\newtheorem{definition}{Definition}
\newtheorem{lemma}{Lemma}

\newtheorem{remark}{Remark}

\newtheorem{assumption}{Assumption}

\newtheorem{proof}{Proof}
\newtheorem{proof of Theorem 1}{Proof of Theorem 1}
\newcommand*{\QEDA}{\hfill\ensuremath{\blacksquare}}

\begin{document}
%

\title{Aggregative games with bilevel structures: Distributed algorithms and\\ convergence analysis}
%
%

\author{Kaihong~Lu,
~ Huanshui Zhang,
~ Long Wang

\thanks{K. Lu and H. Zhang are with College of Electrical Engineering and Automation, Shandong University of Science and Technology,
Qingdao 266590, China 
(e-mails: khong\_lu@163.com; hszhang@sdu.edu.cn).}

 \thanks{L. Wang is with the Center for Systems and Control, College of Engineering, Peking University, Beijing 100871, China (e-mail: longwang@pku.edu.cn)}


}

\maketitle

\begin{abstract}
In this paper, the problem of distributively seeking the equilibria of aggregative games with bilevel structures is studied. Different from the traditional aggregative games, here the aggregation is determined by the minimizer of a virtual leader's objective function in the inner level. Moreover, the global objective function of the virtual leader is formed by the sum of local functions, each of which is determined by the local action of a
 player and the aggregation variable, and is strongly convex with respect to the aggregation variable.
When making decisions, each player only has
access to a local part of the virtual leader's objective function, and can communicate with its neighbors via a connected graph. To handle this problem, first, we propose a second order gradient-based distributed algorithm, where the Hessian matrices associated with the objective functions of the leader are involved. 
Under mild assumptions on the graph and cost functions, we prove that the actions of players asymptotically converge to the Nash equilibrium point. Then, for the case where the Hessian matrices associated with the objective functions of the virtual leader are not
available, we propose a first order gradient-based distributed algorithm, where a distributed estimate strategy
is developed to estimate the gradients of players' cost functions in the outer level. Under the same conditions, we prove that the convergence errors of players' actions to the Nash equilibrium point are linear with respect to the estimate parameters. Finally, simulations are
provided to demonstrate the effectiveness of our theoretical
results.
\end{abstract}
%
\IEEEpeerreviewmaketitle
\section{Introduction}
\IEEEPARstart{A}GGREGATIVE games (AGs) are the noncooperative games where the cost function of each player depends on a common term determined by the actions of all players, and the common term is referred to as an aggregation \cite{Jensen, Ye1}. Distributively seeking the solutions of AGs without full information has received increasing attention in recent years. This is due to its wide practical applications in many areas such as the network congestion control \cite{Alpcan}, the charging
control for large-scaled electric vehicle systems \cite{Shahidehpour}, the Cournot oligopoly market \cite{Algazin}, and real-time energy trading in the smart grid \cite{Y. You}.
\subsection{Literature review}
AGs are a typical kind of noncooperative games. For AGs, the Nash equilibrium (NE) is one
of the most important solution concepts \cite{{Pave2}}-\cite{Lu}. Recently, various results on the problems of distributively seeking NEs have been achieved \cite{Parise2020}-\cite{Liang2017}. In the networked environment, the aggregation is formulated as the average of some local functions, each of which is determined
by the action of a specific player. In this scenario, distributed average tracking strategies are employed for players to estimate the aggregation by communicating with their immediate neighbors \cite{Ye1, Belgioioso2021, Liang2017}. To seek the NE of AGs, a distributed optimal response algorithm is proposed in \cite{Parise2020}, while a decentralized asymmetric projected gradient algorithm is presented in \cite{Paccagnan2019}. Under the unbalanced graph, a continuous-time distributed gradient algorithm is proposed in \cite{Zhu2021}.
For monotone AGs, a distributed Tikhonov regularization algorithm is developed in \cite{Lei2020}. For strongly monotone AGs, the privacy-preserving distributed gradient algorithms are proposed in \cite{Ye2022, Wang2024}. For stochastic AGs, a distributed mirror descent algorithm based on the operator extrapolation strategy is proposed in \cite{Yi2023}. Furthermore, for the AGs with coupled constraints, some primal-dual based distributed algorithms are proposed to seek the generalized NEs in \cite{Belgioioso2021}-\cite{Liang2017}.

All the aforementioned investigations are conducted for AGs with single-level structures. However, AGs with bilevel structures often appear in engineering applications. For example, in the distributed power allocation of
small-cell networks, each small-cell base station aims to minimize its own cost in the outer level. While the costs of all small-cell base stations are influenced by a common price determined by the total cost of the whole power network, which is usually formulated as an objective function of the service provider \cite{Hwang}. For AGs with bilevel structures, the aggregation in the cost functions in the outer level may be determined by the solution of the problem in the inner level, and the equilibrium points in the inner level depend on those in the outer level. These factors bring difficulties in solving the AGs with bilevel structures. In fact, some works on the AG with bilevel structures have been conducted by formulating the problem as a Stackelberg game. The cost functions of the followers are assumed to be quadratic functions with respect to the actions in the outer level and be linear with respect to the decision variable of a leader, centralized algorithms are proposed to seek the Stackelberg equilibrium in \cite{Chen, Maljkovic}, while a semi-decentralized distributed algorithm is presented in \cite{Fabiani}. The aggregations are assumed to be determined by the linear combinations of followers' actions, and the projected gradient algorithm involving real gradient information of cost functions both in the inner level and in the outer level is proposed to seek the Stackelberg equilibrium in \cite{Shokri}.

 It is worth pointing out that, by running the algorithms in \cite{Chen}-\cite{Shokri}, each player needs to have access to the information of the aggregation term. Unfortunately, in practical applications, the aggregation term may rely on the objective function in the inner level and achieving full information from the inner level is usually impossible. For example, in the
small-cell networks, the total cost of the whole network, that decides the price of each station's transmission power, depends on the transmission power strategies of all small-cell base stations \cite{Hwang}. In large-scaled small-cell networks, it is rather expensive, sometimes impossible, for each small-cell base station to achieve the price by gathering the transmission power strategies of all others together. For the AGs with bilevel structures, if the full information associated with the objective function in the inner level is not available, players do not have access to the information
of the aggregation any more.
How to effectively estimate the gradient information of the cost functions in the outer level becomes a challenging problem. These factors motivate the study of this paper.

\subsection{Our contributions}

In this paper, we study the problem of distributively seeking the NE in a new framework of AGs with bilevel structures, whose structure is shown in Fig. \ref{fig-1}. As shown in Fig. \ref{fig-1}, the aggregation in the AGs is determined by the minimizer of a virtual
leader's objective function in the inner level. The global
objective function of the virtual leader is formed by the sum of some
local functions, each of which is determined by the local action of a
 player and the aggregation variable, and is strongly convex with respect to the aggregation variable. In the outer level, players intend to selfishly minimize their own cost functions which are convex with respect to their own actions.
~When making decisions, each player only has access to the information associated with its own cost function, a local part of the virtual
leader's objective function and its own action set, and can communicate with its neighbors via a connected graph.
 The \emph{novelties} and \emph{contributions} of the paper are summarized as follows.
\begin{figure}
\centering
\includegraphics[width=0.46\textwidth]{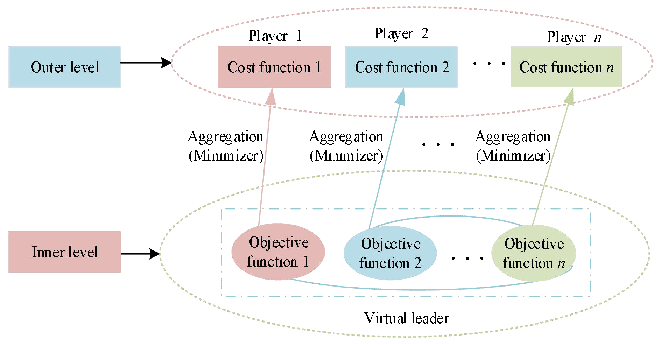}
\caption{The AGs with bilevel structures.} \label{fig-1}
\end{figure}

(i) Compared with AGs with single-level structures \cite{Ye1}, \cite{Parise2020}-\cite{Liang2017}, the problem studied here is more general. More specifically, the formulated AGs with bilevel structures cover the AGs with single-level structures \cite{Ye1}, \cite{Parise2020}-\cite{Liang2017} as well as the distributed optimization problems with strongly convex objective functions \cite{BU1}-\cite{Pu}. Moreover, the problem formulated in this paper is rather different from those studied in \cite{Chen}-\cite{Shokri} where each player has access to the full gradient information of its own cost function. Here the aggregation term, which determines the gradient of each player's cost function in the outer level, relies on the objective function of the virtual
leader in the inner level. However, it is impossible for the players to achieve full information of the objective function in the inner level. It means that, each player in the
outer level does not have access to the full gradient information
of its own cost function.

(ii) To handle this problem, a new second order gradient-based distributed (SOGD) algorithm is proposed based on the primal-dual algorithm, the distributed average tracking strategy and the gradient descent algorithm. Particularly, the distributed average tracking strategy and the gradient descent algorithm are employed to estimate a term determined by the inverse of the Hessian matrix associated with the global objective function.
Different from \cite{Qu} and \cite{Pu}, here the distributed average tracking strategy is employed to track the average of the Hessian matrices associated with the objective functions in the inner level. Moreover, in the proposed algorithm, the projected
gradient strategy involving the Krasnoselskii-Mann iterative
strategy is used to update players' actions in the outer level. While a primal-dual based distributed algorithm involving a fixed step-size is developed for players to estimate the value of aggregation by cooperatively minimizing the global objective function in the inner level. Combining consensus theory, convex analysis theory, and matrix theory, we prove that the actions of players in the outer level asymptotically converge to the NE, and the convergence rate is $\mathcal{O}(\sqrt{\ln t/t})$.

(iii) Considering the case where the Hessian matrices associated with the objective functions in the inner level are not available, we propose a first order gradient-based distributed (FOGD) algorithm to address the problem. Particularly, in the algorithm, a new distributed estimate strategy is developed to estimate the product of the aggregation's gradient and the gradient of the cost function with respect to the second argument in the outer level. By the classical calculus,
the Hessian matrices associated with the objective functions in the inner level are never used when players
make decisions. Using the calculus theory, a linear relationship between the estimated error and the estimate parameter is established. Furthermore, we prove that the convergence errors of players' actions in the outer level to the NE are linear with respect to the estimate parameters.

\subsection{Notations}

Throughout the paper, we use $\mathbb{R}$ and $\mathbb{R}^{m}$
to denote the set of real numbers and $m$-dimensional real vector space, respectively. Given a scalar $x\in\mathbb{R}$, we use $\lceil x\rceil$ to represent the smallest integer that is not smaller than $x$. Given vector ${y}\in\mathbb{R}^{m}$, $\|{y}\|$ denotes the standard Euclidean norm of $y$.  $\langle{u,v}\rangle$ denotes the inner product of $u, v\in\mathbb{R}^{m}$. ${\textbf{1}}_m$ represents an $m$-dimensional column vector whose elements are all $1$. $I_m$ is the $m\times m$ identity matrix.  Given a set $\mathcal{V}=\{1, \cdots, n\}$ and vectors $x_i\in \mathbb{R}^{m}$, where $i\in\mathcal{V}$, we denote $col(x_i)_{i\in\mathcal{V}}=[x_1^T,\cdots,x_n^T]^T$.  The projection onto a set $\Omega$ is denoted by $P_\Omega[\cdot]$, i.e., $P_\Omega[x]=\arg\min_{y\in\Omega}\|y-x\|$.
 Given a function $g(\cdot,\cdot): \mathbb{R}^{m_1}\times \mathbb{R}^{m_2}\rightarrow\mathbb{R}$, we use $\nabla_1g(\cdot,\cdot)$, $\nabla_2g(\cdot,\cdot)$, and $\nabla_{22}g(\cdot,\cdot)$ to represent the gradient of $g$ with respect to the first argument, the gradient with respect to the second argument, and the Hessian
matrix with respect to the second argument, respectively. We denote the Jacobian matrix of $\nabla_2g(\cdot,\cdot)$ with respect to the first argument by $\nabla_{21}g(\cdot,\cdot)$. Given a mapping $\sigma(\cdot): \mathbb{R}^{m_1}\rightarrow\mathbb{R}^{m_2}$, $\nabla_x\sigma(x)\in \mathbb{R}^{m_2\times m_1}$ denotes the Jacobian matrix of $\sigma(x)$ with respect to $x$, where $x\in\mathbb{R}^{m_1}$. Differentiable function $\psi(\cdot): \mathbb{R}^{m}\rightarrow\mathbb{R}$ is said to be convex (respectively, $\mu-$strongly convex) if $\psi(x)-\psi(y)\geq\langle \nabla\psi(y), x-y\rangle$ (respectively, $\psi(x)-\psi(y)\geq\langle \nabla\psi(y), x-y\rangle+\frac{\mu}{2}\| x-y\|^2$) for any $x, y\in\mathbb{R}^m$. Given a mapping $g(\cdot): \mathbb{R}^{m_1}\rightarrow\mathbb{R}^{m_2}$, we say $g(\cdot)$ is $L-$Lipschitz continuous if $\|g(x)-g(y)\|\leq L\|x-y\|$, $\forall~x,y\in\mathbb{R}^{m_1}$ for some constant $L>0$. For a matrix ${A}$, $\|A\|_{\mathbb{F}}$ represents its Frobenius norm, and $[{A}]_{i\cdot}$ represents the $i^{th}$ row. Given a vector $x\in\mathbb{R}^r$ and a matrix $A\in\mathbb{R}^{m\times n}$, we use $x^T$ and $A^T$ to denote their transpositions, respectively. For a mapping $F(\cdot): \mathbb{R}^{m}\rightarrow\mathbb{R}^{n\times r}$, we say $F(\cdot)$ is $L_0-$Lipschitz continuous if each element in $F(x)$ is Lipschitz continuous and $\|F(x)-F(y)\|_{\mathbb{F}}\leq L_0\|x-y\|$, $\forall~x,y\in\mathbb{R}^{m}$ for some constant $L_0>0$. Given two matrices $U$ and $R$, we denote the Kronecker product of them by $U\otimes R$. Given two functions $\phi(\cdot)$, $h(\cdot): \mathcal{R}\rightarrow\mathcal{R}$, $\phi(t)=\mathcal{O}(h(t))$ means $\lim_{t\rightarrow\infty} \frac{\phi(t)}{h(t)}<\infty$.

\section{Problem formulation}
Consider a noncooperative game denoted by $\Gamma(\mathcal{V}, \Omega, {J})$. Here $\mathcal{V}=\{1, 2, \cdots, n\}$ represents the set of players; $\Omega=\Omega_{1}\times\cdots\times\Omega_n$ is the action set of players, where $\Omega_i\subset \mathbb{R}^{m_1}$ is the action set of player $i$; ${J}=\{{J}_{1}, \cdots, {J}_{n}\}$ is the set of players' cost functions, where ${J}_{i}(\cdot, \cdot): \mathbb{R}^{m_1}\times\mathbb{R}^{m_2}\rightarrow \mathbb{R}$ is the cost function of player $i$ at time $t$. For $i\in \mathcal{V}$, $x_i\in \Omega_i$ is the action of player $i$, while $x_{-i}$ denotes the actions of the players other than player $i$, i.e., $x_{-i}=[{x_1}^T,\cdots, {x_{i-1}}^T, {x_{i+1}}^T,\cdots, {x_n}^T]^T$. Moreover, we use $x=(x_i, x_{-i})$ to denote
all players' actions. For any $i\in \mathcal{V}$, player $i$ intends to address the following problem
\begin{equation}\label{a2}\begin{split}
&\min_{{x_i}\in \Omega_i} J_{i}(x_i, \sigma(x))\\
&\emph{\emph{s. t. }}~~\sigma(x)=\arg\min_y \left(\sum_{i=1}^n g_i(x_i,y)\right)
\end{split}
\end{equation}
where ${J}_{i}(\cdot, \cdot): \mathbb{R}^{m_1}\times\mathbb{R}^{m_2}\rightarrow \mathbb{R}$ is differentiable with respect to both arguments, ${J}_{i}(x_i, \sigma(x))$ is convex with respect to $x_i$, $g_i(\cdot, \cdot): \mathbb{R}^{m_1}\times\mathbb{R}^{m_2}\rightarrow \mathbb{R}$ is $\mu_i$-strongly convex and twice differentiable with respect to the second argument for any $x_i\in\Omega_i$. For the sake of simplicity, we denote $\mu=\min_{1\leq i\leq n}\mu_i$. Due to the strong convexity of $g_i(x_i, \cdot)$, for any $x_i\in\Omega_i$, we have the following two equivalent conditions: (i) $\nabla_{22}g_i(x_i, y)-\mu I_{m_2}$ is positive semidefinite; (ii) $\langle\nabla_2g_i(x_i, x)-\nabla_2g_i(x_i, y), x-y\rangle\geq\mu\|x-y\|^2$, $\forall~x,y\in\mathbb{R}^{m_2}$. Game $\Gamma(\mathcal{V}, \Omega, {J})$ is a typical AG with bilevel structures. The one, who holds the objective
function $\sum_{i=1}^n g_i(x_i,y)$ in the inner level, is viewed as a virtual leader. When the players in the outer level make decisions, we assume that each player $i$ only has access to the information associated with its cost function $J_{i}$, a local function $g_i$ in the inner level, and a local set $\Omega_i$. 
In this paper, the main goal is to seek the equilibrium in the outer level. When minimizing the costs in the outer level, the players need to cooperatively estimate the aggregation $\sigma(x)$ by exchanging local information with their neighbors via a network. The network topology is formulated as an undirected graph $\mathcal{G}(\mathcal{V}, E, A)$, where ${\mathcal{V}}$ represents a set of vertices, ${E}\subset{\mathcal{V}}\times{\mathcal{V}}$ denotes an edge set, and ${A}=(a_{ij})_{n\times n}$ is a weighted adjacency matrix satisfying $a_{ij}>0$ if $ (j,i)\in {E}$ and $a_{ij}=0$ otherwise. The following assumptions are made.
\begin{assumption}\label{as1}
$A\textbf{1}_n=\textbf{1}_n$, and $\mathcal{G}(\mathcal{V}, E, A)$ is connected.
 \end{assumption}

 \begin{assumption}\label{as2}
$\Omega_i$ is nonempty, convex and compact for any $i\in\mathcal{V}$, i.e., $\sup_{x_i\in \Omega_i}\|x_i\|\leq K$ for some $K>0$. (Thus, there exists some constants $K_i, i=0, 1, \cdots, 3$ such that $\|\nabla_2J_i(x_i,y)\|\leq K_1$, $\|\nabla_{2}g_i(x_{i},y)\|\leq K_2$, $\|\nabla_{21}g_i(x_i,y)\|_\mathbb{F}\leq K_3$, $\|\nabla_{22}g_i(x_{i},$ $y)\|_\mathbb{F}\leq K_0$ for any $x_i\in \Omega_i$ and $y\in \mathbb{R}^{m_2}$ if $y$ is bounded).
 \end{assumption}

 \begin{assumption}\label{as3}
(Lipschitz continuity) $\nabla_1J_i(x,y)$, $\nabla_2J_i(x,$ $y)$, $\nabla_{2}g_i(x,y)$, $\nabla_{21}g_i(x,y)$ and $\nabla_{22}g_i(x,y)$, $i=1, \cdots, n$, are all $L$-Lipschitz continuous with respect to both the first argument and the second argument for any $x\in \Omega_i$ and $y\in \mathbb{R}^{m_2}$.
 \end{assumption}

 Define the pseudo-gradient mapping ${F}(x)=col({\nabla _{{x_i}}}J_{i}(x_i,$ $\sigma(x)))_{i\in\mathcal{V}}$, where $x=(x_i, x_{-i})$. To guarantee the uniqueness of the equilibrium in the outer level, the following assumption is necessary.
\begin{assumption}\label{as4}(Strong monotonicity) $\langle F(x)-F(y), x-y\rangle\geq \theta \|x-y\|^2, ~\forall x,y\in \Omega$ for some $\theta>0$.
 \end{assumption}
\begin{remark}In (\ref{a2}), if we let $x_i=d_i$ and $f_i(y)=g_i(d_i, y)$ for some constant $d_i \in \Omega_i$, then the action variables $x_i$, $i\in\mathcal{V}$  are dropped in the inner level, accordingly, computing the value of the aggregation becomes a distributed optimization problem with strongly convex objective functions \cite{BU1}-\cite{Pu}. While if letting $g_i(x_i,y)=\frac{\|y\|^2}{2n}-\langle y, \sigma_i(x_i)\rangle$ for some $\sigma_i(\cdot): \mathbb{R}^{m_1}\rightarrow\mathbb{R}^{m_2}$, then we have $\sigma(x)=\sum_{i=1}^n\sigma_i(x_i)$,
accordingly, problem (\ref{a2}) becomes the AG with single-level structure \cite{Ye1}, \cite{Parise2020}-\cite{Liang2017}. Hence, the distributed optimization problem and the traditional AG are two special cases of (\ref{a2}). Consequently, studying (\ref{a2}) benefits for establishing a unified research framework for these two problems. To the best of our knowledge, this is the first time to study the distributed AGs with bilevel structures as (\ref{a2}).
\end{remark}

Due to the strong convexity of $g_i$ with respect to the second argument, the function $\sigma(x)$ is uniquely determined, so in (\ref{a2}), the cost functions of players in the outer level are single-valued mappings with respect to $x$, i.e., $J_i(x_i, \sigma(x)):\mathbb{R}^{nm_1}\rightarrow\mathbb{R}$. The equilibrium point of the players' actions in the outer level is an NE, which is defined as follows.
\begin{definition}
In problem (\ref{a2}), the action profile $x^*=(x_i^*, x_{-i}^*)\in \Omega$ is an NE if it satisfies $J_{i}(x_i^*, \sigma(x^*))\leq J_{i}(x_i, \sigma(x_i, x_{-i}^*))$ for any $x_i\in\Omega_i$ and $i\in\mathcal{V}$.
\end{definition}
\begin{remark}
By (\ref{a2}), if $x^*=(x_i^*, x_{-i}^*)$ is the NE defined in Definition 1, then we have $\sigma(x^*)=\arg\min_y (\sum_{i=1}^n g_i(x_i^*,$ $y))$. If we view the one who holds the objective function in the inner level as a leader, and view the players who hold the cost functions in the outer level as followers, then the problem can be formulated as a Stackelberg game. Letting $\Omega_i=\mathbb{R}^{m_1}$, we know that the first order optimality condition for the equilibrium in Definition 1 implies
\begin{equation*}\label{a2091-bb}\left\{\begin{split}
& \nabla_{x_i}J_{i}(x^*_i, \sigma(x^*))=\nabla_{1}J_{i}(x^*_i, \sigma(x^*))+\\
&~~~~~~~~~(\nabla_{x_i} \sigma(x^*))^T\nabla_2 J_i(x_i^*, \sigma(x^*)) =0,~i\in\mathcal{V}\\
&\sum_{i=1}^n \nabla_2g_i(x_i^*,\sigma(x^*))=0.
\end{split}\right.
\end{equation*}
Since the gradient $\nabla_{x_i}J_{i}$ in the outer level depends on the term $\nabla_{x_i} \sigma(x)$, we need to establish the relationship between $\nabla_{x_i} \sigma(x)$ and $g_i$. In \cite{Chen}-\cite{Shokri}, the Stackelberg game where $\sigma(x)$ does not rely on the best response in the inner level is studied. Different from them, here we consider the case where the aggregation $\sigma(x)$ relies on the minimizer of the objective function in the inner level, and each player only has access to a local part of the objective function.
\end{remark}

Based on Definition 1 and results in \cite{Chen, Facchinei}, a sufficient and necessary condition associated with the NEs of AG with bilevel structures is provided.
\begin{lemma}[\cite{Chen, Facchinei}]\label{LE000}
For game $\Gamma(\mathcal{V}, \Omega, {J})$, if $\Omega$ is convex and compact, and $J_i(x_i,$ $\sigma(x))$, $i\in\mathcal{V}$ is differentiable and convex with respect to $x_i$, then the NEs always exist. Moreover,
$x^*$ is an NE if and only if for any $k>0$,
\begin{equation}\label{eq-00002}\begin{split}
&x^*-P_\Omega[x^*-k{F}(x^*)]=0.\\
\end{split}\end{equation}
\end{lemma}

Based on Lemma \ref{LE000}, we know that under Assumption 2, the solution to (2) exists. In the following proposition, each component of ${F}(x)$ is provided.
\begin{proposition} Given ${J}_{i}(\cdot, \cdot): \mathbb{R}^{m_1}\times\mathbb{R}^{m_2}\rightarrow \mathbb{R}$ and $g_i(\cdot,$ $\cdot): \mathbb{R}^{m_1}\times\mathbb{R}^{m_2}\rightarrow \mathbb{R}$ as (\ref{a2}),\\
(i) $\nabla_{x_i} \sigma(x)=-\left(\sum_{i=1}^n \nabla_{22}g_i(x_i, \sigma(x))\right)^{-1}\nabla_{21}g_i(x_i, \sigma(x))$;\\
(ii)${\nabla _{{x_i}}}J_{i}(x_i,$ $\sigma(x))=\nabla_1 J_i(x_i,\sigma(x))+(\nabla_{x_i} \sigma(x))^T\nabla_2 J_i(x_i,$ $\sigma(x))$.
 \end{proposition}
\begin{proof}See APPENDIX. A for details. \end{proof}

Due to the strong convexity of $g_i(x_i, \cdot)$ and the boundedness of $\nabla_{21}g_i(x_i, \sigma(x))$, based on (i) in Proposition 1 and Assumption 2, there holds that $\nabla_{x_i} \sigma(x)$ is bounded, then it is not difficulty to verify that
\begin{equation}\label{104-22}\begin{split}
\|\sigma(x)-\sigma(y)\|\leq\ell\|x-y\|, \forall~x,y\in\Omega
\end{split}
\end{equation}
where $\ell=\frac{K_3}{\sqrt{n}\mu}$.

Under Assumption \ref{as3}, we can conclude that $F(x)$ is $L_0$-Lipschitz continuous for any $x\in \Omega$.
\begin{lemma}\label{LE0-00}
Under Assumptions \ref{as2} and \ref{as3}, $F(x)$ is $L_0$-Lipschitz continuous for any $x\in \Omega$, where $L_0=nL(K_2(K_3$ $+\mu +K_3\ell+\mu \ell)+(K_3 \mu+\mu^2) (1+\ell))/\mu^2$.
\end{lemma}
\begin{proof}See APPENDIX. B for details.\end{proof}

Assumptions \ref{as2} and
 \ref{as4} ensure the existence and uniqueness of the solution to (\ref{eq-00002}), denoted by $x^*=(x_i^*, x_{-i}^*)$. Moreover, due to the strong convexity of $\sum_{i=1}^ng_i(x_i^*, \cdot)$ for any $x_i^*\in\Omega_i$, $\sigma(x^*)$ is the unique solution to $\min_{y}\sum_{i=1}^ng_i(x_i^*, y)$. In fact, the Lipschitz continuity and strong monotonicity of the pseudo-gradient $F(x)$ are commonly adopted in the study of games, even in the cases with single-level structures \cite{Ye1}, \cite{Parise2020}-\cite{Liang2017}, \cite{Facchinei}-\cite{Lu10}.
\begin{remark}
In traditional AGs \cite{Ye1}, \cite{Parise2020}-\cite{Liang2017}, the aggregation is defined as $\sigma(x)=\sum_{i=1}^n\sigma_i(x_i)$, where $\sigma_i(\cdot): \mathbb{R}^{m_1}\rightarrow\mathbb{R}^{m_2}$. It is easy to achieve that $\nabla_{x_i} \sigma(x)=\nabla_{x_i} \sigma_i(x_i)$, so the components of $\nabla_{x_i} \sigma(x)$ are not coupled at all, and the information of $\nabla_{x_i} \sigma_i(x_i)$ is available to player $i$. However, by the fact that $\sum_{i=1}^n \nabla_{22}g_i(x_i, \sigma(x))$ is symmetric and positive definite, and using Proposition 1, we have
\begin{equation}\label{a5001}\begin{split}
&{\nabla _{{x_i}}}J_{i}(x_i, \sigma(x))=\nabla_1 J_i(x_i,\sigma(x))-(\nabla_{21}g_i(x_i, \sigma(x)))^T\\
&~~~~\times \left(\sum_{i=1}^n \nabla_{22}g_i(x_i, \sigma(x))\right)^{-1}\nabla_2 J_i(x_i, \sigma(x)).
\end{split}
\end{equation}
In (\ref{a5001}), $\sigma(x)$ is determined by the distributed optimization problem in the inner level of (\ref{a2}), and it can not be directly accessed. Moreover, all ${\nabla _{{x_i}}}J_{i}(x_i, \sigma(x))$ and all $\nabla_{x_i} \sigma(x)$, $i=1, \cdots, n$ are coupled through the common term $\left(\sum_{i=1}^n \nabla_{22}g_i(x_i, \sigma(x))\right)^{-1}$, which can not be directly accessed by the players. When making decisions, how to estimate the value of ${\nabla _{{x_i}}}J_{i}$ for players only using local information is a challenging problem. This problem formulation is rather different from the cases studied in \cite{Fabiani}-\cite{Shokri} where the information of ${\nabla _{{x_i}}}J_{i}$ is available to player $i$.
\end{remark}
\section{Distributed methods involving second order gradient information}\label{se2}
\subsection{Algorithm design}
Denote the action of player $i$ at any time $t\geq0$ by $x_{i, t}\in\mathbb{R}^{m_1}$, where $i\in\mathcal{V}$. To estimate $\sigma(x_t)$ defined in the inner level of (\ref{a2}), it is necessary to solve the optimization problem $\min_{y\in\mathbb{R}^{m_2}}\sum_{i=1}^n g_i(x_{i, t},y)$. 
 Define $\tilde{y}=col(y_i)_{i\in\mathcal{V}}$, where $y_{i}\in\mathbb{R}^{m_2}$. Under Assumption 1, we know that $y_1=y_2=\cdots=y_n$ is equivalent to $((I-A)\otimes I_{m_2})\tilde{y}=0$ \cite{Olfati}. Accordingly, the problem $\min_{y\in\mathbb{R}^{m_2}}\sum_{i=1}^n g_i(x_{i, t},y)$ is equivalent to
\begin{equation}\label{-a2}\begin{split}
&\min_{\tilde{y}\in \mathbb{R}^{nm_2}} g(x_t, \tilde{y})=\sum_{i=1}^n g_i(x_{i,t},y_i)\\
&\emph{\emph{s. t. }}~~ ((I-A)\otimes I_{m_2})\tilde{y}=0.
\end{split}
\end{equation}
 To address this problem, we propose the following primal-dual based distributed algorithm
\begin{equation}\label{s1}
 \left\{
\begin{split}
y_{i,t+1}&=y_{i,t}+\kappa\Big(\sum_{j \in {\mathcal{N}_i}}a_{ij}\zeta_{j,t}-\zeta_{i,t}-\nabla_2g_i(x_{i,t},y_{i,t})\Big)\\
\zeta_{i,t+1}&=\zeta_{i,t}-\kappa\left(\sum_{j \in {\mathcal{N}_i}}a_{ij}y_{j,t+1}-y_{i,t+1}\right)
\end{split}
\right.
\end{equation}
where $y_{i, t}\in\mathbb{R}^{m_2}$ represents the estimate on the aggregation $\sigma(x_t)$, $\zeta_{i,t}\in\mathbb{R}^{m_2}$ is the dual variable, $\mathcal{N}_i$ is the neighbor set of player $i$, and $\kappa$ is a fixed step-size satisfying $0<\kappa<\min\Big(\frac{1}{L}, $  $\frac{\mu}{(\lambda_n(I-A))^2}\Big)$, $\lambda_n(I-A)$ is the largest eigenvalue of $I-A$. Moreover, to minimize the player's cost function in the outer level, the information of gradient ${\nabla _{{x_i}}}J_{i}(x_i, \sigma(x))$ is required. Based on (\ref{a5001}), to compute ${\nabla _{{x_i}}}J_{i}(x_i, \sigma(x))$, it is necessary to compute $-\left(\sum_{i=1}^n \nabla_{22}g_i(x_{i,t}, y_{i,t})\right)^{-1}\nabla_2 J_i(x_{i,t}, y_{i,t})$, which is a solution of the following optimization problem
\begin{equation}\label{a5003}\begin{split}
&\min_{z_i\in \mathbb{R}^{m_2}} \frac{z_i^T \left(\sum_{j=1}^n \nabla_{22}g_j(x_{j,t}, y_{j,t})\right)z_i}{2}+z_i^T\nabla_2 J_i(x_{i,t}, y_{i,t}).
\end{split}
\end{equation}
So estimating the term $-\left(\sum_{i=1}^n \nabla_{22}g_i(x_{i,t}, y_{i,t})\right)^{-1}\nabla_2 J_i(x_{i,t},$ $ y_{i,t})$ is equivalent to solving (\ref{a5003}). It is not difficult to verify that the gradient of objective function in (\ref{a5003}) is $(\sum_{j=1}^n \nabla_{22}g_j(x_{j,t}, y_{j,t}))z_i+\nabla_2 J_i(x_{i,t}, y_{i,t})$, where the second order gradient information associated with $g_i$ is involved. To estimate $(\sum_{j=1}^n \nabla_{22}g_j(x_{j,t},$ $ y_{j,t}))^{-1}\nabla_2 J_i(x_{i,t}, y_{i,t})$, we can first employ distributed average tracking algorithm (\ref{s2}) to track $\sum_{j=1}^n \nabla_{22}g_j(x_{j,t}, y_{j,t})$, and then use the gradient descent algorithm (\ref{s3}) to solve (\ref{a5003}). Note that when tracking the term $\sum_{j=1}^n \nabla_{22}g_j(x_{j,t}, y_{j,t})$ in a distributed way, the current information of players' actions and the estimate on the aggregation is involved, so the distributed average tracking algorithm for estimating the term $\sum_{j=1}^n \nabla_{22}g_j(x_{j,t}, y_{j,t})$ should be implemented after the estimate on the aggregation and the players' actions update. \\
  \noindent
\rule[0\baselineskip]{8.9cm}{1pt}
\emph{Algorithm 1: SOGD algorithm}\\
\rule[0.53\baselineskip]{8.9cm}{1pt}
\textbullet~\textbf{Initialization:} At each iteration time $t\geq0$, each player $i$ for any $i\in\mathcal{V}$
maintains its action $x_{i, t}\in\mathbb{R}^{m_1}$, the estimate $y_{i, t}\in\mathbb{R}^{m_2}$ on the aggregation $\sigma(x_t)$, the dual variable $\zeta_{i,t}\in\mathbb{R}^{m_2}$, the estimate $v_{i,t}\in \mathbb{R}^{m_2\times m_2}$ on the average of $\nabla_{22}g_i(\cdot, \cdot)$, and the estimate $z_{i,t}\in \mathbb{R}^{m_2}$ on $-\left(\sum_{i=1}^n \nabla_{22}g_i(x_i, y_{i,t})\right)^{-1}\nabla_2 J_i(x_i,$ $y_{i,t})$.
Set the initial values as
$x_{i, 0} \in \Omega_i$, $y_{i, 0}\in \mathbb{R}^{m_2}$, $\zeta_{i,0}=0$,  $z_{i, 0}\in \mathbb{R}^{m_2}$, $v_{i, 0}=\nabla_{22}g_i(x_{i, 0}, y_{i, 0})$.\\
\textbullet~\textbf{Iteration: } For $t=0, 1, 2, \cdots$ and $i\in\mathcal{V}$, each player $i$ updates variables using the following rules.

\textbullet~Update variables $y_{i,t}$ and $\zeta_{i,t}$ in the inner level as (\ref{s1}).

\textbullet~Update action $x_{i,t}$ in the outer level by the following projected gradient strategy
\begin{equation}\label{s4}
 \left\{
\begin{split}
&x_{i,t+1}=(1-\eta_t)x_{i,t}+\eta_tP_{\Omega_i}\big[x_{i,t}-k\hat{F}_{i,t}\big]\\
&\hat{F}_{i,t}=\nabla_1J_i(x_{i,t},y_{i,t})+(\nabla_{21}g_i(x_{i,t},y_{i,t}))^Tz_{i,t}
\end{split}
\right.
\end{equation}
where $\eta_t$ is a nonincreasing step-size such that $0\leq\eta_t\leq 1$, and $\frac{\theta}{L_0}> k>0$ is constant.

~\textbullet~Update $v_{i,t}$ by the following distributed average tracking algorithm
\begin{equation}\label{s2}
\begin{split}
v_{i,t+1}=&\sum_{j \in {\mathcal{N}_i}}a_{ij}v_{j,t}+\nabla_{22}g_i(x_{i,t+1},y_{i,t+1})\\
&-\nabla_{22}g_i(x_{i,t},y_{i,t}).
\end{split}
\end{equation}

~\textbullet~Update $z_{i,t}$ by the following gradient descent algorithm
\begin{equation}\label{s3}
\begin{split}
&z_{i,t+1}=z_{i,t}-\alpha(nv_{i,t}z_{i,t}+\nabla_2J_i(x_{i,t},y_{i,t}))\\
\end{split}
\end{equation}
where $\alpha$ is a fixed step-size such that $0<\alpha<\frac{\mu}{2nK_0^2}$.\\
\rule[0.3\baselineskip]{8.9cm}{1pt}

To handle problem (\ref{a2}), now we propose the SOGD algorithm, see Algorithm 1 for details. Note that by running Algorithm 1, each player $i$ updates its
  action only using local action information received from its neighbors, its own action set $\Omega_i$, and the information associated with a local function $g_i$ in the inner level and its own cost function $J_{i}$ in the outer level. Thus, Algorithm 1 is distributed.
\begin{remark}
In Algorithm 1, the design of dynamics (\ref{s1}) is motivated by the primal-dual algorithms in \cite{Feijer}-\cite{Hamedani}. The design of dynamics (\ref{s2}) is motivated by the distributed average tracking algorithm \cite{Qu, Pu}, as well as the consensus algorithms \cite{Olfati, Ren}. Different from the distributed algorithms for minimizing the objective function only with single argument \cite{Feijer}-\cite{Hamedani}, here we need to deal with the functions with two arguments which are inseparable. The convergence of (\ref{s1}) relies on the actions $x_{i,t}$, $i\in\mathcal{V}$ of the players. If the value of $x_{i,t}$ changes fast, then (\ref{s1}) could not converge. To reduce the fluctuation of $x_{i,t}$, the projected gradient strategy involving the Krasnoselskii-Mann iterative strategy is proposed in (\ref{s4}).
\end{remark}

\subsection{Convergence results}

For simplicity, we denote
\begin{equation}\label{-a0}\begin{split}
B=I-A.
\end{split}
\end{equation}
Note that if the graph is connected, then $B$ has a simple zero-eigenvalue, and the others are all positive \cite{Olfati}-\cite{Barshooi}. The eigenvalues can be presented in an ascending order: $0=\lambda_1(B)<\lambda_2(B)\leq\lambda_3(B)\leq\cdots\leq\lambda_n(B)$. For (\ref{-a2}), the Lagrangian function can be defined as $$L^t(\tilde{y}, \lambda)=g(x_t, \tilde{y})+\zeta^T(B\otimes I_{m_2})\tilde{y}$$ where $\zeta\in\mathbb{R}^{nm_2}$. Based on the Karush-Kuhn-Tucker (KKT) condition \cite{Yang}, for any $t\geq0$, $y^*_t=\left[\sigma^T(x_t), \cdots, \sigma^T(x_t)\right]^T$ is the solution of (\ref{-a2}) if and only if
\begin{equation}\label{-a3}\left\{
\begin{split}
&\nabla_2 g(x_t, {y}^*_t)+(B\otimes I_{m_2})\zeta^*_t=0\\
&(B\otimes I_{m_2}){y}^*_t=0
\end{split}
\right.
\end{equation}
for some $\zeta^*_t=col(\zeta^*_{i, t})_{i=1, \cdots, n}$. In the following lemma, we establish the relationship between the variation in sequence $\{\zeta^*_t\}_{t=0, 1, \cdots}$ and that in sequence $\{\nabla_2 g(x_t, {y}^*_t)\}_{t=0, 1, \cdots}$.
\begin{lemma}\label{LE-1}
Under Assumption 1, there exists a solution sequence $\{\zeta^*_t\}_{t=0, 1, \cdots}$ to (\ref{-a3}) satisfying\\
(i)$ \|(B\otimes I_{m_2})(\zeta^*_{t}-\zeta)\|\geq \lambda_2(B)\|\zeta^*_{t}-\zeta\|$ for any $\zeta\in Null(\textbf{1}_n^T\otimes I_{m_2})$,\\
(ii)$ \|\zeta^*_{t+1}-\zeta^*_{t}\|\leq \frac{1}{\lambda_2(B)}\|\nabla_2 g(x_{t+1}, {y}^*_{t+1})-\nabla_2 g(x_t, {y}^*_t)  \|$,\\
 where $B$ is defined in (\ref{-a0}).
\end{lemma}
\begin{proof} See APPENDIX. C for details.\end{proof}

Then, by (\ref{s1}) and (\ref{s4}), the estimated error between ${y}_{i,t}$ and the aggregation $\sigma(x_t)$ is provided, refer to the following lemma for details.
\begin{lemma}\label{LE2}
Under Assumptions \ref{as1}-\ref{as3}, by (\ref{s1}) and (\ref{s4}), for some $0<\gamma<1$,
\begin{equation}\label{-eq103}\begin{split}
\|y_{i,t}-\sigma(x_t)\|\leq \mathcal{O}\left(\gamma^t+\sum_{k=1}^t\gamma^{t-k}\eta_{k-1}\right),~~i\in\mathcal{V}. \\
\end{split}\end{equation}
\end{lemma}
\begin{proof} See APPENDIX. D for details.
\end{proof}

Based on Lemma \ref{LE2}, by Assumption 2, we know that $y_i(t)$ is bounded. Moreover, if $\lim_{t\rightarrow\infty}\eta_t=0$, then $y_{i,t}$ converges to $\sigma(x_t)$ for any $i\in\mathcal{V}$. Thus, the aggregation is estimated by all players. Particularly, based on the proof of Lemma \ref{LE2}, we know that the bound of $\|y_{i,t}-\sigma(x_t)\|$ is influenced by the variation of actions, i.e., $\|x_{i,t+1}-x_{i,t}\|$, which results in the term $\sum_{k=1}^t\gamma^{t-k}\eta_{k-1}$ in (\ref{-eq103}) by using (\ref{s4}). If we define $f_i$ as Remark 1, and reduce problem (\ref{a2}) to be a distributed optimization problem, then the term $\|x_{i,t+1}-x_{i,t}\|$ is zero. Accordingly, the term $\sum_{k=1}^t\gamma^{t-k}\eta_{k-1}$ in (\ref{-eq103}) does not exist. Then the convergence rate of $y_{i,t}$ to the minimizer is improved to be linear, i.e.,  $\|y_{i,t}-\sigma^*\|\leq \mathcal{O}\left(\gamma^t\right)$, where $\sigma^*=\arg\min_{y}f_i(y)$. By (\ref{-eq103}), we know that $y_{i,t}$ is bounded for any $t\geq0$. Next, we will provide the estimated error of $v_{i,t}$ on the average of Hessian matrices $\nabla_{22}g_i(x_{i,t},y_{i,t})$, $i\in\mathcal{V}$.
\begin{lemma}\label{LE4}
Under Assumptions \ref{as1}-\ref{as3}, for some $0<\tau<1$,
\begin{equation*}\begin{split}
\|v_{i,t}-\bar{G}_{t}\|_\mathbb{F}&\leq \mathcal{O}\Big(\sum_{k=0}^t\tau^{t-k}\eta_{k}+(t+1)\tau^{t}\\
&~~~~+\sum_{k=0}^t\tau^{t-k}\sum_{s=0}^k\tau^{k-s}\eta_{s}\Big),~~i\in\mathcal{V}
\end{split}\end{equation*}
where $\bar{G}_t=\frac{\sum_{i=1}^n\nabla_{22}g_i(x_{i,t},y_{i,t})}{n}$.
\end{lemma}
\begin{proof}
 See APPENDIX. E for details.
\end{proof}

By (\ref{a5001}), to achieve the gradient information in the outer level, it is necessary to compute the term $-(\sum_{i=1}^n\nabla_{22}g_i(x_{i, t}, $ $\sigma(x_t)))^{-1}\nabla_2J_i(x_{i,t},\sigma(x_t))$. For simplicity, we denote
\begin{equation}\label{eq11-26}
h_{i, t}=-\left(\sum_{i=1}^n\nabla_{22}g_i(x_{i,t},\sigma(x_t))\right)^{-1}\nabla_2J_i(x_{i,t},\sigma(x_t)).
\end{equation}
In Algorithm 1, $z_{i,t}$ is the estimate on $-(\sum_{i=1}^n\nabla_{22}g_i(x_{i,t}, $ $y_{i,t}))^{-1}\nabla_2J_i(x_{i,t},y_{i,t})$, while $y_{i,t}$ is the estimate on $\sigma(x_t)$. 
In what follows, a preliminary result on the difference between $z_{i,t}$ and the $h_{i, t}$ will be provided.
\begin{lemma}\label{LE5}
Under Assumptions \ref{as1}-\ref{as3}, for any $i\in \mathcal{V}$,
\begin{equation}\label{eq116}\begin{split}
  &\|z_{i,t+1}-h_{i, t+1}\|\\
  &\leq\Big(\vartheta+r_1\| \bar{G}_t-v_{i,t}\|_\mathbb{F}\Big)\|z_{i,t}-h_{i, t}\|\\
   &~~~+\frac{4LK\ell\eta_t}{\mu}+r_2\| \bar{G}_t-v_{i,t}\|_\mathbb{F}\\
   &~~~+r_3\Big(2+\vartheta+r_1\| \bar{G}_t-v_{i,t}\|_\mathbb{F}\Big)\sum_{i=1}^n\|y_{i,t}-\sigma(x_t)\|\\
  \end{split}\end{equation}
where $h_{i, t}$ is defined in (\ref{eq11-26}), $\vartheta=\sqrt{1-{n\mu\alpha}+2n^2K_0^2\alpha^2}<1$, $r_1=\sqrt{\frac{2\alpha n}{\mu}+4n^2\alpha^2}$, $r_2=$ $\sqrt{\frac{2\alpha K_1^2}{n\mu^3}+\frac{4K_1^2}{\mu^2}}$, and $r_3=\frac{nK_1L}{\mu^2}+\frac{nL}{\mu}$.
\end{lemma}
\begin{proof}See APPENDIX. F for details.\end{proof}

Now we present the convergence result on Algorithm 1 in the following theorem.
\begin{theorem}
Under Assumptions \ref{as1}-\ref{as4}, if $\eta_{t}=\frac{b}{t+a}$ for some $a>\max (2k(\theta-kL_0)b, b)$ and $b\geq \max \left(\frac{1}{2k(\theta-kL_0)}, 1\right)$, then by Algorithm 1, both $\|x_{t}-x^*\|$ and $\|y_{i,t}-\sigma(x^*)\|$, $\forall~i\in\mathcal{V}$ converge to zero with the convergence rate $\mathcal{O}(\sqrt{\ln t/t})$, where $x^*$ is defined in Definition 1.
\end{theorem}
\begin{proof}See APPENDIX. G for details.\end{proof}

\begin{remark}Based on Theorem 1, we know that the convergence rate of Algorithm 1 follows $\mathcal{O}(\sqrt{\ln t/t})$, which is faster than the common convergence rate $\mathcal{O}(\ln t/\sqrt{t})$ of distributed optimization algorithms with decaying step-size \cite{Duchi, Olshevsky}. Although the Krasnoselskii-Mann iterative strategy in (\ref{eq2001}) slows the convergence rate of the projected gradient algorithm, the fixed step-sizes $\alpha, \kappa$ and $k$ help to improve the convergence rate of the whole algorithm.\end{remark}

\section{Distributed methods involving first order gradient information}\label{Se2}
Note that when running Algorithm 1, it is necessary to compute the Hessian matrix. However, in many practical applications such as meta learning systems \cite{1AA}, the loss functions are usually high-dimensional mappings, so computing or estimating the
second order gradient information is rather expensive, sometimes impossible. This factor motivates the developments of distributed algorithms involving first order gradient information.
\subsection{Algorithm design}
To avoid the usage of Hessian matrices in Algorithm 1, now we directly estimate the term $(\nabla_{x_i} \sigma(x_t))^T\nabla_2 J_i(x_{i, t},\sigma(x_t))$. To do this, we define an auxiliary variable as follows.
\begin{equation}\label{eq113001}\begin{split}
D_{i,t}=\frac{\nabla_1 g_i(x_{i,t}, y_{i}(\delta_i,x_t))-\nabla_1 g_i(x_{i,t}, \sigma(x_t))}{\delta_i}, i\in\mathcal{V}
\end{split}\end{equation}
where $y_{i}(\delta_i,x_t)$ is determined by the following auxiliary optimization problem
\begin{equation}\label{eq2001}\begin{split}
y_{i}(\delta_i,x_t)=\arg\min_{y_i}\delta_iJ_i(x_{i, t},y_i)+\sum_{j=1}^ng_j(x_{j, t}, y_i)
\end{split}\end{equation}
for some $\delta_i>0$. $D_{i,t}$ can be viewed as an estimate on the term $(\nabla_{x_i} \sigma(x_t))^T\nabla_2 J_i(x_{i, t},\sigma(x_t))$. Particularly, if $\delta_i=0$, then $y_{i}(0,x_t)=\sigma(x_t)$, where $\sigma(x)$ defined in (\ref{a2}). By (\ref{eq2001}), we know that global information is involved in $y_{i}(\delta_i,x_t))$, we need to further estimate this variable. Before going on, the following assumption is presented.
\begin{assumption}\label{as5}
For any $x_i\in\Omega_i$ and $i\in\mathcal{V}$, $J_i(x_i,\cdot)$ is convex. Moreover, $\|\nabla_{22}J_i(x_{i}, y_i)\|_\mathbb{F}\leq\hbar$, $y_i\in\mathbb{R}^{m_2}$ for some $\hbar>0$.
 \end{assumption}

Since in game $\Gamma(\mathcal{V}, \Omega, {J})$, each player's cost is determined by the aggregation $\sigma(x)$, then $J_i$ may not be convex with respect to $x_{-i}$        under Assumption 5. Moreover, by the strong convexity of $g_i(x_{i, t},\cdot)$, and using Assumption 5, we know that $\delta_iJ_i(x_{i, t},y_i)+\sum_{j=1}^ng_j(x_{j, t}, y_i)$ is $\mu-$strongly convex with respect to the second argument. Accordingly, $y_{i,t}(\delta_i,x_t)$ is the unique solution of (\ref{eq2001}). Due to the differences of $y_i(\delta_i, x_t)$ from $y_k(\delta_k, x_t)$ for any $i, k\in\mathcal{V}$, here we employ $w_{ik,t}\in \mathbb{R}^{m_2}$ to represent the estimate of player $k$ on $y_i(\delta_i, x_t)$. To estimate $y_i(\delta_i, x_t)$ for each player, the following distributed primal-dual algorithm is proposed
\begin{equation}\label{eq2002}
 \left\{\begin{split}
&w_{ik,t+1}=w_{ik,t}+\beta_i\Big(\sum_{j \in {\mathcal{N}_k}}a_{kj}u_{ij,t}-u_{ik,t}\\
&~~~~~~-\nabla_2g_k(x_{k,t},w_{ik,t})-c_{ki}\delta_i\nabla_2J_i(x_{i, t},w_{ii,t})\Big)\\
&u_{ik,t+1}=u_{i,t}-\beta_i\Big(\sum_{j \in {\mathcal{N}_k}}a_{kj}w_{ij,t+1}-w_{ik,t+1}\Big)
\end{split}
\right.
\end{equation}
where $c_{k i}=1$ if $i=k$, $c_{ki}=0$ otherwise, $u_{ik, t}\in  \mathbb{R}^{m_2}$ is the dual variable with initial value $u_{ik, 0}=0$, and $\beta_i>0$ is the step-size.
To handle problem (\ref{a2}), now we propose the FOGD algorithm, see Algorithm 2 for details. In Algorithm 2, we use $d_{i,t}$ to estimate $(\nabla_{x_i} \sigma(x_t))^T\nabla_2 J_i(x_{i, t},$ $ \sigma(x_t))$ by the estimate strategy (\ref{s33}). By Algorithm 2, the Hessian matrix information of the objective functions in the inner level is not required any more. Thus, Algorithm 2 is applicable to practical applications where computing the second order gradient information is rather expensive or impossible. 
\\
\noindent
\rule[0\baselineskip]{8.9cm}{1pt}
\emph{Algorithm 2: FOGD algorithm}\\
\rule[0.53\baselineskip]{8.9cm}{1pt}
\textbullet~\textbf{Initialization:} At each iteration time $t\geq0$, each player $i$ for any $i\in\mathcal{V}$
maintains its action $x_{i, t}\in\mathbb{R}^{m_1}$, the estimate $y_{i, t}\in\mathbb{R}^{m_2}$ on $\sigma(x_t)$ on the aggregation, and the estimate $w_{ik,t}$ on $y_{i}(\delta_i,x_t)$.
Set the initial values as
$x_{i, 0} $ $\in \Omega_i$, $y_{i, 0}\in \mathbb{R}^{m_2}$,  $w_{i, 0}\in \mathbb{R}^{m_2}$.\\
\textbullet~\textbf{Iteration: } For any $t=0, 1, \cdots$ and $i\in\mathcal{V}$, each player $i$ updates variables using the following rules.

\textbullet~Update variables $y_{i,t}$ and {\color{blue}$\zeta_{i,t}$} in the inner level as (\ref{s1}).

\textbullet~Update $w_{ik,t}$ as (\ref{eq2002}) with $0<\beta_i<\min\Big(\frac{1}{L+\hbar\delta_i},$ $ \frac{\mu}{(\lambda_n(I-A))^2}\Big)$.

\textbullet~Compute the estimate on the gradient
\begin{equation}\label{s33}
 \left\{
\begin{split}
&\hat{F}_{i,t}=\nabla_1J_i(x_{i,t},y_{i,t})+d_{i,t}\\
&d_{i,t}=\frac{\nabla_1 g_i(x_{i,t}, w_{ii,t})-\nabla_1 g_i(x_{i,t}, y_{i,t})}{\delta_i}
\end{split}
\right.
\end{equation}
where $\delta_i>0$ is the estimate parameter.

\textbullet~Update action $x_{i,t}$ in the outer level by the following projected gradient strategy
\begin{equation}\label{s44}
\begin{split}
&x_{i,t+1}=(1-\eta_t)x_{i,t}+\eta_tP_{\Omega_i}\big[x_{i,t}-k\hat{F}_{i,t}\big]\\
\end{split}
\end{equation}
where $\eta_t$ is a nonincreasing step-size such that $0\leq\eta_t\leq 1$, and $\frac{\theta}{L_0}> k>0$ is constant.\\
\rule[0.3\baselineskip]{8.9cm}{1pt}

\subsection{Convergence results}
In this section, we will provide the convergence analysis of Algorithm 2. First, in the next lemma, the relationship between $D_{i,t}$ and $(\nabla_{x_i} \sigma(x_t))^T\nabla_2 J_i(x_{i, t},\sigma(x_t))$ is established.
\begin{lemma}\label{LE10}
Under Assumptions 2, 3 and 5,  for any $\delta_i>0$ and $i\in\mathcal{V}$,
\begin{equation*}\label{eq3001}\begin{split}
\|D_{i,t}-(\nabla_{x_i} \sigma(x_t))^T\nabla_2 J_i(x_{i, t},\sigma(x_t))\|\leq\mathcal{C}\delta_i
\end{split}\end{equation*}
where $\mathcal{C}=\frac{LK_1^2}{2n\mu^2}+\frac{\mathcal{C}_0K_2}{2}$, $\mathcal{C}_0=\frac{\sqrt{n}K_1^2L(1+\hbar)+\sqrt{n}\mu K_1L}{\mu^3}$, and $D_{i,t}$ is defined in $(\ref{eq113001})$.
\end{lemma}
\begin{proof}
See APPENDIX. H for details.
 \end{proof}

Note that if $w_{ii, t}$ and $y_{i, t}$ approach to $y_{i}(\delta_i,x_t)$ and $\sigma(x_t)$, respectively, then $d_{i,t}$ approximates $D_{i,t}$. Based on Lemma 7, we can conclude that the estimated error is linear with the estimate parameter $\delta_i$. With the help of the lemmas above, now we provide the convergence result on Algorithm 2.
\begin{theorem}
Under Assumptions \ref{as1}-\ref{as5}, if $\eta_{t}=\frac{b}{t+a}$ for some $a>\max (2k(\theta-kL_0)b, b)$ and $b\geq \max \left(\frac{1}{2k(\theta-kL_0)}, 1\right)$, then by Algorithm 2, for any $\forall~i\in\mathcal{V}$,
\begin{equation}\label{3120}\begin{split}
\|x_{t}-x^*\|\leq\mathcal{O}\left(\sqrt{\ln t/{t}}+\sum_{i=1}^n(\sqrt{\mathcal{C}\delta_i}+\mathcal{C}\delta_i)\right)
\end{split}\end{equation}
and
\begin{equation}\label{3121}\begin{split}
\|y_{t}-\sigma(x^*)\|\leq\mathcal{O}\left(\sqrt{\ln t/{t}}+\sum_{i=1}^n(\sqrt{\mathcal{C}\delta_i}+\mathcal{C}\delta_i)\right)
\end{split}\end{equation}where $\mathcal{C}$ is defined in Lemma \ref{LE10}, and $x^*$ is defined in Definition 1..
\end{theorem}
\begin{proof}See APPENDIX. I for details.\end{proof}

\begin{remark} In (\ref{3120}) and (\ref{3121}), the term $\sqrt{\mathcal{C}\delta_i}+\mathcal{C}\delta_i$ comes from the estimate errors of $d_{i,t}$ defined in (\ref{s33}) on the value of $\nabla_{x_i} \sigma(x_t)\nabla_2 J_i(x_{i, t},\sigma(x_t))$. For any $i\in\mathcal{V}$, $\delta_i$ is a significant variable that is used to estimate the term $\nabla_{x_i} \sigma(x_t)\nabla_2 J_i(x_{i, t},\sigma(x_t))$. On the one hand, by (\ref{3120}) and (\ref{3121}), smaller value of $\delta_i$ implies more accurate convergence. On the other hand, to ensure (\ref{eq113001}) and (\ref{s33}) being well defined, the value of $\delta_i$ can not be too small. How to select an optimal value of $\delta_i$ is still a difficult problem, which will be further studied in our future work by establishing some new performance indicators.
\end{remark}

\section{A numerical simulation example}\label{se3}

\begin{figure}
\centering
\includegraphics[width=0.2\textwidth]{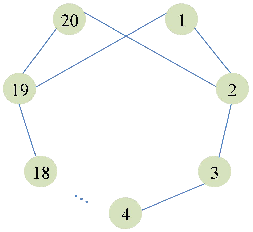}
\caption{The communication graph.} \label{fig0}
\end{figure}
In this section, we use our algorithms to deal with the distributed power allocation problem of small-cell networks. Here we consider a small-cell network consisting of 20 small-cells, labeled as $\mathcal{V}=\{1, \cdots, 20\}$. And each small-cell has a local small-cell base station,
which can cover many users. The small-cell base stations are viewed as players. Here we use $x_i\in\mathbb{R}$ to denote the transmission power strategy of small cell base station $i$, the range of the transmission power strategy is given by $x_i\in[0, P_{i, 0}]$, i.e. $\Omega_i=\{x_i|0\leq x_i\leq P_{i, 0}\}$. We model the power allocation problem as a game with bilevel structures \cite{Hwang}. Each small-cell base station
aims to adjust its transmission power to minimize its cost as follows
    \begin{equation}\label{9121}\begin{split}
    J_{i}(x_i, y)=a_iyx_i-b_i(1+c_ix_i),~i\in\mathcal{V}
     \end{split}\end{equation}
where $a_i$, $b_i$ and $c_i$ are some constants determined by the power transmission conditions of small-cell base station $i$. Moreover, $y\in\mathbb{R}$ represents the price of the power, which can not be directly accessed by the small cell and is determined by the minimization of a total cost as follows.
\begin{equation}\label{9122}\begin{split}
\sum_{i=1}^{20}\left(d_iy^2-a_ix_iy\right)
     \end{split}\end{equation}
where $d_i$ is a local parameter used to adjust the price by small-cell base station $i$ based on the power consumption of the users in the areas covered by small-cell base station $i$. The small-cell base stations cover different users, whose power requirements or consumptions are often different. Accordingly, the values of $d_i$, $i=1, \cdots, n$ are different. The costs in (\ref{9121}) and (\ref{9122}) can be viewed as the costs in the outer level and the objective function in the inner level, respectively. Using the full information, we can obtain that $\sigma(x)=\frac{\sum_{i=1}^{20}a_ix_i}{2\sum_{i=1}^{20}d_i}$, which is an aggregation. 
  In the large-scaled network, it is usually impossible for small-cells to use the full information. When making decisions, here we assume that each small-cell $i$ can only access the information of $a_iyx_i-b_i(1+c_ix_i)$ and $d_iy^2-a_ix_iy$. And each small-cell $i$ can communicate with its neighbors via the connected graph shown in Fig. \ref{fig0}. In this formulation, the parameters are given by $a_1=a_7=a_{9}=a_{11}=a_{19}=2.5$, $a_2=a_{12}=3$, $a_3=a_{13}=a_{17}=1.5$, $a_4=a_{14}=2$, $a_{5}=a_{8}=a_{10}=a_{15}=a_{18}=a_{20}=4$, $a_{6}=a_{16}=1$, $b_i =3,$ $c_i=a_i$, $d_1=d_3=d_7=d_{11}=d_{13}=d_{17}=0.1,$ $ d_2=0.05,~d_4=d_{14}=0.09,$  $d_5=d_{15}=0.2,$  $ d_6=d_{16}$  $=0.02$, $ d_8=d_{18}=0.07,$ $ d_9=d_{19}=0.095,$ $ d_{10}$  $=d_{20}=0.13,$ and $P_{i, 0}=0.9$ for any $i\in\mathcal{V}$. Accordingly, the NE follows that $x^*=[0.22, 0.18, 0.36, 0.27, $ $0.14, 0.55, 0.22, 0.14, 0.22, 0.14, 0.22, 0.18, 0.36, 0.27, 0.14, $ $0.55, 0.36, 0.14, 0.22, 0.14]^T$, which implies that $\sigma(x^*)=$ $2.86$.

Now we solve the power allocation problem by using Algorithm 1. Each small-cell $i$ maintains an estimate on its own transmission power strategy, denoted by $x_{i, t}\in\mathbb{R}$,
as well as a local estimate on the price of
the transmission power, denoted by $y_{i, t}\in\mathbb{R}$. We select the step-sizes as $\alpha=0.01, k=\kappa=1$ and $\eta_t=\frac{3}{t+4}$. And initial values of $x_{i, t}$ and $z_{i, t}$ are randomly chosen from $0$ to $1.5$, while initial values of $y_{i, t}$ are randomly chosen from $0$ to $2$ for all $i\in \mathcal{V}$. Let $E(t)=\sum_{i=1}^{20}(x_i(t)-x_i^*)^2$ and $y^*=\sigma(x^*)$, by running Algorithm 1, the trajectories of $E(t)$ and $y_{i, t}-y^*$ are shown in Fig. \ref{fig1} and Fig. \ref{fig2}, respectively. From Fig. \ref{fig1} and Fig. \ref{fig2}, we see that after a period of time, $(x_t, y_{i, t})$ approaches to $(x^*, \sigma(x^*))$ for any $i\in\mathcal{V}$. Thus, the simulations are consistent with the theoretical results established in
Theorem 1.

Next, we use Algorithm 2 to deal with the same problem by selecting $k=\kappa=0.8$, $\beta_i=1$, $\delta_1=\delta_7=\delta_{12}=\delta_{13}=\delta_{16}=1.3,~\delta_2=\cdots=\delta_6=\delta_{10}=\delta_{11}=\delta_{15}=\delta_{17}=\cdots=\delta_{19}=1$, $\delta_9=\delta_{14}=1.5$, $\delta_8=\delta_{20}=1.6$, and $\eta_t=\frac{0.12}{t+1}$. Under Algorithm 2, the trajectories of $E(t)$ and $y_{i, t}-y^*$ are shown in Fig. \ref{fig3} and Fig. \ref{fig4}, respectively. From Fig. \ref{fig3} and Fig. \ref{fig4}, we see that $(x_t, y_{i, t})$ can also keep close to $(x^*, \sigma(x^*))$ for any $i\in\mathcal{V}$. The simulations are consistent with the theoretical results established in
Theorem 2.

\begin{figure}
\begin{minipage}[t]{0.46\linewidth}
\centering
\includegraphics[width=1\textwidth]{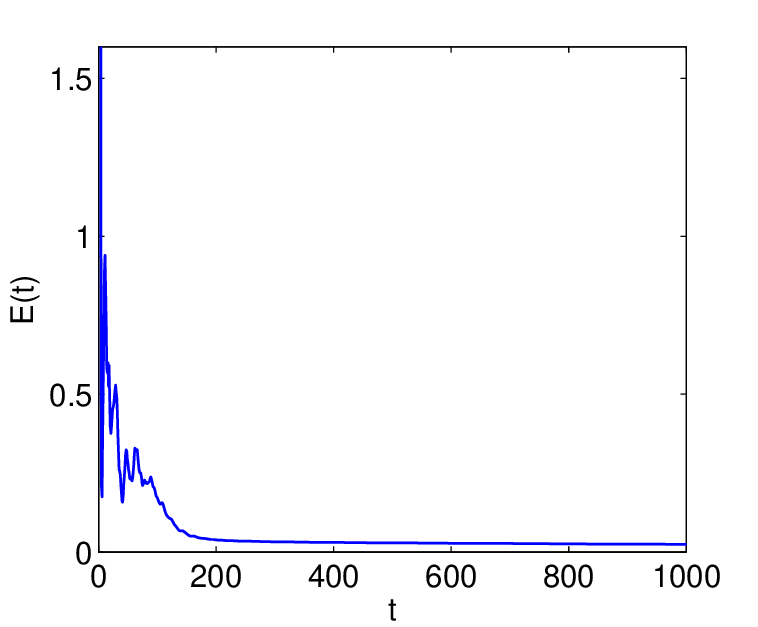}
\caption{The trajectory of $E(t)$ under Algorithm 1.}\label{fig1}
\end{minipage}
\begin{minipage}[t]{0.46\linewidth}
\centering
\includegraphics[width=1\textwidth]{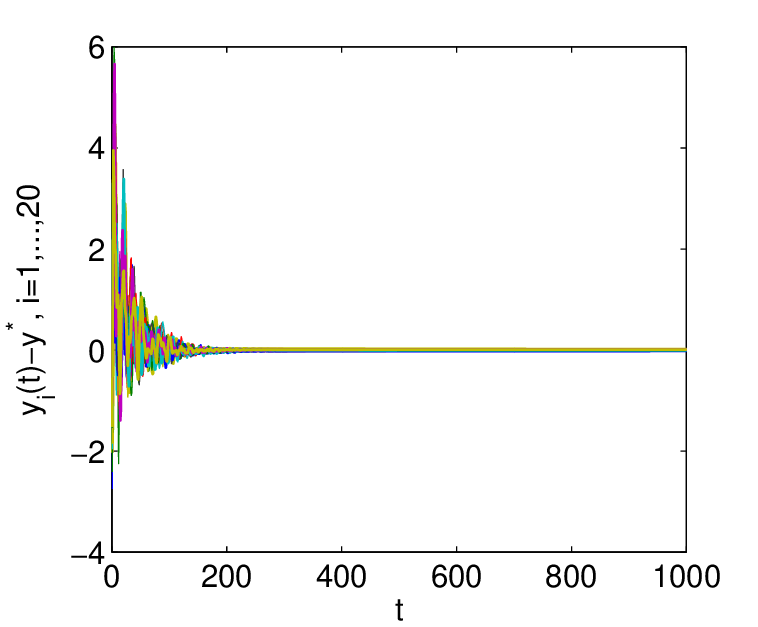}
\caption{The trajectories of $y_i(t)-y^*,~i$ $=1,~\cdots,~20$ under Algorithm 1.}\label{fig2}
\end{minipage}
\end{figure}
\begin{figure}
\begin{minipage}[t]{0.48\linewidth}
\centering
\includegraphics[width=1\textwidth]{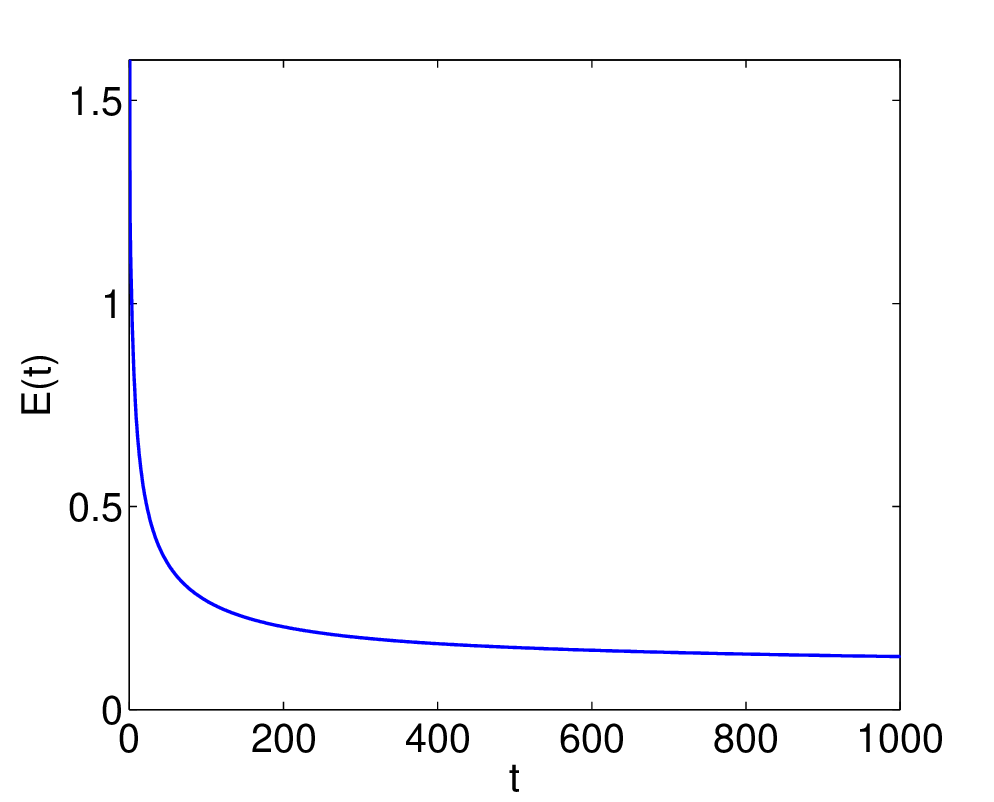}
\caption{The trajectory of $E(t)$ under Algorithm 2.}\label{fig3}
\end{minipage}
\begin{minipage}[t]{0.48\linewidth}
\centering
\includegraphics[width=1\textwidth]{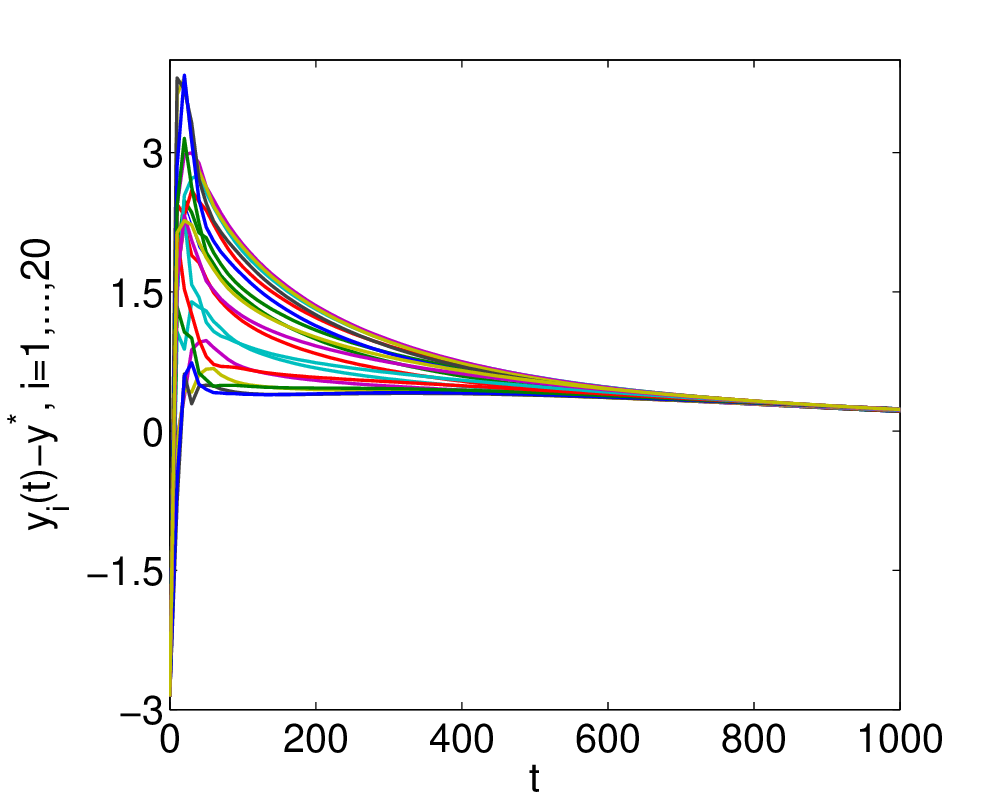}
\caption{The trajectories of $y_i(t)-y^*,~i$ $=1,~\cdots,~20$ under Algorithm 2.}\label{fig4}
\end{minipage}
\end{figure}

\section{Conclusions}\label{se4}
In this paper, we have studied the problem of distributively seeking the NEs of the AGs with bilevel structures, where the aggregation is determined by the minimizer of the global objective function of a virtual leader in the inner level. To handle this problem, first, an SOGD algorithm involving the Hessian matrices associated with the objective functions in the inner level has been proposed. By implementing the algorithm, players in the
outer level make decisions only using
the information associated with its own cost function, a local part of the virtual leader's objective function, its own action set and its own action, as well as the actions
of its neighbors. The result shows that if the graph is connected, then the actions of players asymptotically converge to the NE. Furthermore, for the case where the Hessian matrices associated with the objective functions in the inner level are not
available, we have proposed an FOGD algorithm, where a distributed estimate strategy
is developed to estimate the gradients of cost functions in the outer level. The result shows that the convergence errors of players' actions to the NE are linear with respect to the estimate parameters.

Based on the result in Theorem 2, using the distributed algorithm involving the first order gradient information only can achieve an approximate solution. Currently, achieving an accurate convergence only by using the first order gradient information is still a difficult problem, which will be further studied. In our future work, some other issues on the AGs with bilevel structures will be considered, such as the case with network induced time delays, packet loss and communication bandwidth constraints, which
will bring new challenges in seeking the NEs in distributed manners.

\section{Appendix}
In this section, we will provide the proofs
of our results in detail. First, let us begin with presenting the
proof of Proposition 1.
\subsection{Proof of Proposition 1}
(i) By the fact that $\sigma(x)=\arg\min_y\sum_{i=1}^n g_i(x_i, y),$ there holds that $\sum_{i=1}^n\nabla_2g_i(x_i, \sigma(x))=0$. Then taking the Jacobian matrix with respect to $x_i$ and using the chain derivative rule yield
\begin{equation}\label{104}
\begin{split}
\nabla_{21}g_i(x_i, \sigma(x))+\left(\sum_{i=1}^n \nabla_{22}g_i(x_i, \sigma(x))\right)\nabla_{x_i} \sigma(x)=0.
\end{split}
\end{equation}
Since $\sum_{i=1}^ng_i(x_i,\sigma(x))$ is strongly convex, then we know that $\sum_{i=1}^n \nabla_{22}g_i(x_i, \sigma(x))$ is symmetric and positive definite. Accordingly, (\ref{104}) implies (i).

(ii) Using the chain derivative rule, we have \begin{equation*}
\begin{split}{\nabla _{{x_i}}}J_{i}(x_i,\sigma(x))=&\nabla_1 J_i(x_i,\sigma(x))\\
&+(\nabla_{x_i} \sigma(x))^T\nabla_2 J_i(x_i,\sigma(x)).\end{split}
\end{equation*} Hence, the validity of (ii) is verified.\QEDA
\subsection{Proof of Lemma 2}
Using Assumption 2 and the strong convexity of $g_i$ to (i) in proposition 1, we know that $\|\nabla_{x_i} \sigma(x)\|_\mathbb{F}\leq\ell$ for any $x, y\in\Omega$. Moreover, by the continuity of $\sigma(\cdot)$, there holds that $\sigma(x)$ and $\sigma(y)$ are bounded for any $x, y\in\Omega$. Then, we have
\begin{equation*}\begin{split}
&\|\nabla_{x_i} \sigma(x)-\nabla_{x_i} \sigma(y)\|_\mathbb{F}\\
&= \Big\|\Big(\sum_{i=1}^n \nabla_{22}g_i(x_i, \sigma(x))\Big)^{-1}\nabla_{21}g_i(x_i, \sigma(x))-\\
&~~~\Big(\sum_{i=1}^n \nabla_{22}g_i(y_i, \sigma(y))\Big)^{-1}\nabla_{21}g_i(y_i, \sigma(y))\Big\|_\mathbb{F}\\
&\leq K_3\Big\|\Big(\sum_{i=1}^n \nabla_{22}g_i(x_i, \sigma(x))\Big)^{-1}\Big\|_\mathbb{F}\Big\|\sum_{i=1}^n \nabla_{22}g_i(x_i, \sigma(x))\\
&~~~-\sum_{i=1}^n \nabla_{22}g_i(y_i, \sigma(y))\Big\|_\mathbb{F}\Big\|\Big(\sum_{i=1}^n \nabla_{22}g_i(y_i, \sigma(y))\Big)^{-1} \Big\|_\mathbb{F}\\
&~~~+\Big\| \Big(\sum_{i=1}^n \nabla_{22}g_i(y_i, \sigma(y))\Big)^{-1}\Big\|_\mathbb{F}\|\nabla_{21}g_i(x_i, \sigma(x))\\
&~~~-\nabla_{21}g_i(y_i, \sigma(y))\|_\mathbb{F}\\
&\leq \left(\frac{K_3L}{\sqrt{n}\mu^2}+\frac{L}{\mu}+\frac{K_3L\ell}{\mu^2}+\frac{L\ell}{\sqrt{n}\mu}\right)\|x-y\|.\\
\end{split}
\end{equation*}
Note that
\begin{equation*}\begin{split}
&\|(\nabla_{x_i} \sigma(x))^T\nabla_2 J_i(x_i,\sigma(x))-(\nabla_{x_i} \sigma(y))^T\nabla_2 J_i(y_i,\sigma(y))\|\\
&\leq\|\nabla_{x_i} \sigma(x)\|_\mathbb{F}\|\nabla_2 J_i(x_i,\sigma(x))-\nabla_2 J_i(y_i,\sigma(y))\|\\
&~~~+\|\nabla_{x_i} \sigma(x)-\nabla_{x_i} \sigma(y)\|_\mathbb{F}\|\nabla_2 J_i(y_i,\sigma(y))\|\\
&\leq\frac{K_3}{\sqrt{n}\mu}\|\nabla_2 J_i(x_i,\sigma(x))-\nabla_2 J_i(y_i,\sigma(y))\|\\
\end{split}
\end{equation*}
\begin{equation}\label{as2222}\begin{split}
&~~~+K_2\|\nabla_{x_i} \sigma(x)-\nabla_{x_i} \sigma(y)\|_\mathbb{F}\\
&\leq\frac{K_3}{\sqrt{n}\mu} (L+L\ell)\|x-y\|\\
&~~~+K_2\|\nabla_{x_i} \sigma(x)-\nabla_{x_i} \sigma(y)\|_\mathbb{F}\\
&\leq L_1\|x-y\|
\end{split}
\end{equation}
where $L_1=\frac{K_2L(K_3+\sqrt{n}\mu +\sqrt{n}K_3\ell+\mu \ell)+K_3 \mu L(1+\ell)}{\sqrt{n}\mu^2}$. Moreover,
\begin{equation}\label{as2223}\begin{split}
&\|\nabla_1 J_i(x_i,\sigma(x))-\nabla_1 J_i(y_i,\sigma(y))\|\\
&\leq (1+\ell)L\|x-y\|.
\end{split}
\end{equation}
Based on (\ref{as2222}) and (\ref{as2223}), we can achieve that
\begin{equation*}\begin{split}
\|F(x)-F(y)\|&\leq \sum_{i=1}^n\|\nabla_1 J_i(x_i,\sigma(x))-\nabla_1 J_i(y_i,\sigma(y))\|\\
&~~~+\sum_{i=1}^n\|(\nabla_{x_i} \sigma(x))^T\nabla_2 J_i(x_i,\sigma(x))\\
&~~~-(\nabla_{x_i} \sigma(y))^T\nabla_2 J_i(y_i,\sigma(y))\| \\
&\leq  L_0\|x-y\|.
\end{split}
\end{equation*}
Hence, the validity of the result is verified.\QEDA
\subsection{Proof of Lemma 3}
(i) Note that for the linear equations in the first line of (\ref{-a3}), there must exist a solution $\zeta^*_t$ belongs to the range space of $B\otimes I_{m_2}$, i.e., there exists some $v\in\mathbb{R}^{nm_2}$ such that $(B\otimes I_{m_2})v=\zeta^*_t$. Since $(\textbf{1}_n^T\otimes I_{m_2})(B\otimes I_{m_2})v=0$, then $\zeta^*_t\in Null(\textbf{1}_n^T\otimes I_{m_2})$. Together with the fact that $\zeta\in Null(\textbf{1}_n^T\otimes I_{m_2})$, we have $\zeta^*_{t}-\zeta\in Null(\textbf{1}_n^T\otimes I_{m_2})$, i.e., $ (\textbf{1}_n^T\otimes I_{m_2})(\zeta^*_{t}-\zeta)=0$. Note that $\textbf{1}_n^TB=0$ and $B^2$ is symmetric, so there exists an orthogonal matrix $U=[\textbf{1}_n/\sqrt{n}, U_0]^T$ satisfying
\begin{equation*}\begin{split}
(\zeta^*_{t}-\zeta)^T(B\otimes I_{m_2})^2(\zeta^*_{t}-\zeta)&= (\zeta^*_{t}-\zeta)^T U \Lambda^2 U^T(\zeta^*_{t}-\zeta)\\
&\geq (\lambda_2(B))^2\|\zeta^*_{t}-\zeta\|^2
 \end{split}
 \end{equation*}
 where $\Lambda=diag\{0, \lambda_2(B), \cdots, \lambda_n(B)\}$.
Thus, the validity of (i) is verified.

(ii) By the arbitrariness of $t\geq0$, we have
\begin{equation}\label{-a4}\left\{
\begin{split}
&\nabla_2 g(x_t, {y}^*_t)+(B\otimes I_{m_2})\zeta^*_t=0\\
&\nabla_2 g(x_{t+1}, {y}^*_{t+1})+(B\otimes I_{m_2})\zeta^*_{t+1}=0.
\end{split}
\right.
\end{equation}
By (\ref{-a4}), there holds that $$\|(B\otimes I_{m_2})(\zeta^*_{t+1}-\zeta^*_{t})\|=\|\nabla_2 g(x_t, {y}^*_t)-\nabla_2 g(x_{t+1}, {y}^*_{t+1})\|.$$
Using (i) and letting $\zeta=\zeta^*_{t+1}$ in (i), it immediately leads to the validity of (ii).\QEDA
\subsection{Proof of Lemma 4}
Denote $y_t=col(y_{i,t})_{i\in\mathcal{V}}$, $\zeta_t=col(\zeta_{i,t})_{i\in\mathcal{V}}$, $y^*_t=\left[\sigma^T(x_t), \cdots, \sigma^T(x_t)\right]^T$, and let $\zeta^*_t$ be defined in Lemma \ref{LE-1}, by (\ref{s1}) in Algorithm 1 and KKT condition in (\ref{-a3}), we have
\begin{equation}\label{9}
 \left\{
\begin{split}
&y_{t+1}-y^*_t=y_t-y^*_t-\kappa((B\otimes I_{m_2})(\zeta_t-\zeta^*_t)\\
&~~~~~~~~~~~~~~~+\nabla_2g(x_t,y_t)-\nabla_2g(x_t,y^*_t))\\
&\zeta_{t+1}-\zeta^*_t=\zeta_t-\zeta^*_t+\kappa((B\otimes I_{m_2})(y_{t+1}-y^*_t)).
\end{split}
\right.
\end{equation}
Taking the square at both sides of the equations in (\ref{9}) yields that
\begin{equation}\label{10}\begin{split}
&\|y_{t+1}-y^*_t\|^2\\
&=\|y_{t}-y^*_t-\kappa(\nabla_2g(x_t,y_t)-\nabla_2g(x_t,y_t^*))\|^2\\
&~~~~-2\kappa\langle (B\otimes I_{m_2})(\zeta_t-\zeta_t^*), y_{t}-y^*_t\\
&~~~~-\kappa(\nabla_2g(x_t,y_t)-\nabla_2g(x_t,y_t^*))\rangle\\
&~~~~+\kappa^2\|(B\otimes I_{m_2})(\zeta_t-\zeta_t^*)\|^2
\end{split}\end{equation}
and
\begin{equation}\label{10-1}\begin{split}
&\|\zeta_{t+1}-\zeta^*_t\|^2\\
&=\|\zeta_{t}-\zeta^*_t\|^2+\kappa^2\|(B\otimes I_{m_2})(y_{t+1}-y^*_t)\|^2\\
&~~~~+2\kappa\langle\zeta_{t}-\zeta^*_t, (B\otimes I_{m_2})(y_{t+1}-y^*_t)\rangle\\
&= \|\zeta_{t}-\zeta^*_t\|^2+\kappa^2\|(B\otimes I_{m_2})(y_{t+1}-y^*_t)\|^2\\
&~~~-2\kappa^2\|(B\otimes I_{m_2})(\zeta_t-\zeta^*_t)\|^2+\\
&~~~2\kappa\langle (B\otimes I_{m_2})(\zeta_t-\zeta_t^*), y_{t}-y^*_t\\
&~~~-\kappa(\nabla_2g(x_t,y_t)-\nabla_2g(x_t,y_t^*))\rangle.
\end{split}\end{equation}
Combining (\ref{10}) and (\ref{10-1}) results in that
\begin{equation}\label{11}\begin{split}
&\|y_{t+1}-y^*_t\|^2+\|\zeta_{t+1}-\zeta^*_t\|^2\\
&=\|y_{t}-y^*_t-\kappa(\nabla_2g(x_t,y_t)-\nabla_2g(x_t,y_t^*))\|^2\\
&~~~+\kappa^2\|(B\otimes I_{m_2})(\zeta_t-\zeta_t^*)\|^2+\|\zeta_{t}-\zeta^*_t\|^2\\
&~~~+\kappa^2\|(B\otimes I_{m_2})(y_{t+1}-y^*_t)\|^2\\
&~~~-2\kappa^2\|(B\otimes I_{m_2})(\zeta_t-\zeta^*_t)\|^2\\
&\leq \|y_{t}-y^*_t-\kappa(\nabla_2g(x_t,y_t)-\nabla_2g(x_t,y_t^*))\|^2\\
&~~~+\kappa^2(\lambda_n(B))^2\|y_{t+1}-y^*_t\|^2\\
&~~~+(1-\kappa^2 (\lambda_2(B))^2)\|\zeta_{t}-\zeta^*_t\|^2\\
\end{split}\end{equation}
where the inequality holds by the fact that $\|B(\zeta_t-\zeta^*_t)\|\leq \lambda_n(B)\|\zeta_{t}-\zeta^*_t\|^2$ and using (i) in Lemma \ref{LE-1}.
By (\ref{s4}), we know that $x_i(t)\in\Omega_i$ for any $i\in{\mathcal{V}}$. Using the Lipschitz
continuity and the strong monotonicity of $\nabla_2g_i(\cdot, \cdot)$ with respect to the second argument yields that
\begin{equation}\label{11-1}\begin{split}
&\|y_{t}-y^*_t-\kappa(\nabla_2g(x_t,y_t)-\nabla_2g(x_t,y_t^*))\|^2\\
&\leq \|y_{t}-y^*_t\|^2+\kappa^2L\langle y_{t}-y^*_t, \nabla_2g(x_t,y_t)-\nabla_2g(x_t,y_t^*)\rangle\\
&~~~-2\kappa\langle y_{t}-y^*_t, \nabla_2g(x_t,y_t)-\nabla_2g(x_t,y_t^*)\rangle\\
&\leq \|y_{t}-y^*_t\|^2+(\kappa^2L-2\kappa)\mu \|y_{t}-y^*_t\|^2\\
&=(1-\kappa\mu(1-\eta_yL))(1-\kappa^2(\lambda_n(B))^2)\|y_{t}-y^*_t\|^2\\
&~~~-\kappa(\mu-\kappa(\lambda_n(B))^2(1-\kappa\mu(1-\eta_yL)))\|y_{t}-y^*_t\|^2\\
&\leq (1-\kappa\mu(1-\eta_yL))(1-\kappa^2(\lambda_n(B))^2)\|y_{t}-y^*_t\|^2\\
\end{split}\end{equation}
where the last inequality is true due to the facts that $\kappa< \frac{1}{L}$ and $\kappa< \frac{\mu}{(\lambda_n(B))^2}$.
From (\ref{11}) and (\ref{11-1}), it follows that
\begin{equation}\label{eq100}\begin{split}
&(1-\kappa^2(\lambda_n(B))^2)\|y_{t+1}-y^*_t\|^2+\|\zeta_{t+1}-\zeta^*_t\|^2\\
&\leq \gamma_0 ((1-\kappa^2(\lambda_n(B))^2)\|y_{t}-y^*_t\|^2+\|\zeta_{t}-\zeta^*_t\|^2)\\
\end{split}\end{equation}
where $\gamma_0=\max(1-\kappa\mu(1-\kappa L),1-\kappa^2 (\lambda_2(B))^2))$. By the fact that $\mu/L\leq 1$, it is not difficult to verify that $\gamma_0\in[0, 1)$. Letting $e_t=\left[\sqrt{1-\kappa^2(\lambda_n(B))^2}(y_{t}-y^*_t)^T, (\zeta_{t}-\zeta^*_t)^T\right]^T$ and $\gamma=\sqrt{\gamma_0}$, by (\ref{eq100}), we have
\begin{equation}\label{eq101}\begin{split}
\|e_{t+1}\|\leq \gamma \|e_t\|+ \|r_t\| \\
\end{split}\end{equation}
where $r_t=\Big[\sqrt{1-\kappa^2(\lambda_n(B))^2}(y_{t+1}^*-y^*_{t})^T, (\zeta_{t+1}^*-\zeta^*_t)^T\Big]^T$. Using (ii) in Lemma \ref{LE-1} yields
\begin{equation*}\begin{split}
\|\zeta^*_{t+1}-\zeta^*_{t}\|&\leq \frac{1}{\lambda_2(B)}\|\nabla_2 g(x_{t+1}, {y}^*_{t+1})-\nabla_2 g(x_t, {y}^*_t)  \|\\
&\leq \frac{L}{\lambda_2(B)} (\|x_{t+1}-x_t\|+\|{y}^*_{t+1}-{y}^*_t\|)
\end{split}\end{equation*}
where the second inequality holds by using the Lipschitz
continuity of $\nabla_2g_i(\cdot, \cdot)$. Then letting $\varrho_1=\frac{L}{\lambda_2(B)}$ and $\varrho_2=\sqrt{1-\kappa^2(\lambda_n(B))^2}+\varrho_1$, it is easy to verify that
\begin{equation}\label{eq102}\begin{split}
\|r_t\|&\leq \varrho_1\|x_{t+1}-x_t\|+\varrho_2\|{y}^*_{t+1}-{y}^*_t\|.
\end{split}\end{equation}
By (\ref{eq101}) and (\ref{eq102}), we have
\begin{equation}\label{eq103}\begin{split}
\|e_{t+1}\|\leq \gamma \|e_t\|+ \varrho_1\|x_{t+1}-x_t\|+\varrho_2\|{y}^*_{t+1}-{y}^*_t\|. \\
\end{split}\end{equation}
Due to the fact that $\sum_{i=1}^n g_i(x_i, \sigma(x))$ is strongly convex, one knows that $\sum_{i=1}^n \nabla_{22}g_i(x_i, \sigma(x))$ is bounded. By
(\ref{104-22}), there holds $\|{y}^*_{t+1}-{y}^*_t\|=\sqrt{n}\|\sigma(x_{t+1})-\sigma(x_{t})\|\leq\sqrt{n}\ell\|x_{t+1}-x_t\|$. Moreover, by (\ref{s4}), we have $\|x_{t+1}-x_t\|\leq 2\sqrt{n}K\eta_t$, where $K$ is defined in Assumption 2. By (\ref{eq103}), letting $\varrho_0=2\sqrt{n}K(\varrho_1+\sqrt{n}\varrho_1\ell)$, we have
\begin{equation*}\begin{split}
\|e_{t+1}\|\leq \gamma \|e_t\|+ \varrho_0\eta_t. \\
\end{split}
\end{equation*}
It follows that $e_{t+1}\leq \gamma^{t} \|e_0\|+\varrho_0\sum_{k=1}^t\gamma^{t-k}\eta_{k-1}$, which immediately implies (\ref{-eq103}). Thus, the validity of the result is verified.\QEDA
\subsection{Proof of Lemma 5}
Note that if the graph is connected, then $A$ has a simple 1 eigenvalue, and the absolute values of the other eigenvalues are less than 1 \cite{Olfati, Duchi}. So the absolute values of the eigenvalues can be presented in an ascending order: $\rho_1(A)\leq\cdots\leq\rho_{n-1}(A)<\rho_n(A)=1$. To prove Lemma 5, we present the following lemma.
\begin{lemma}[\cite{Duchi}]\label{LE1}
Under Assumption \ref{as1}, for any $i,j\in\mathcal{V}$,
\begin{equation}\label{1}\begin{split}
\left|[A^t]_{ij}-\frac{1}{n}\right|\leq \sqrt{n}\rho^{t}
\end{split}\end{equation}where $\rho=\rho_{n-1}(A)$.
\end{lemma}

\textbf{Proof of Lemma 5.} Denote $\bar{v}_t=\frac{1}{n}\sum_{i=1}^nv_{i,t}$, taking the average at both sides of (\ref{s2}) in Algorithm 1, we have $\bar{v}_{k+1}=\bar{v}_{k}+\bar{G}_{k+1}-\bar{G}_{k}$.
Taking the summation with respect
to $k=0,\cdots, t$ yields $\bar{v}_{t+1}-\bar{v}_0=\bar{G}_{t+1}-\bar{G}_{0}$.
Due to the fact that $\bar{v}_0=\bar{G}_{0}$, there holds $\bar{v}_t=\bar{G}_{t}$. Moreover, denote ${v}_t=col(v_{i,t})_{i\in\mathcal{V}}$, $d_{i,t}=\nabla_{22}g_i(x_{i,t},y_{i,t})-\nabla_{22}g_i(x_{i,t-1},y_{i,t-1})$ and $d_t=col(d_{i, t})_{i\in\mathcal{V}}$, we have
\begin{equation}\label{eq105}\begin{split}
v_{i,t}&=\sum_{j\in \mathcal{N}_{i}}a_{ij,t-1}v_{j,t-1}+d_{i,t}\\
&=\Big([A]_{i\cdot}\otimes I_{m_2}\Big)v_{t-1}+d_{i,t}\\
&=\Big([A^{t}]_{i\cdot}\otimes I_{m_2}\Big)v_0+\sum_{k=1}^t\Big([A^{t-k}]_{i\cdot}\otimes I_{m_2}\Big)d_{k}.
\end{split}\end{equation}
By the fact that $A^t$ is a doubly stochastic matrix, one has
\begin{equation}\label{eq106}\begin{split}
\bar{v}_{t}=\Big(\frac{\textbf{1}_n^T}{n}\otimes I_{m_2}\Big)v_0+\sum_{k=1}^t\Big(\frac{\textbf{1}_n^T}{n}\otimes I_{m_2}\Big)d_{k}.
\end{split}\end{equation}
Combining (\ref{eq105}) and (\ref{eq106}) results in that for any $i\in \mathcal{V}$,
\begin{eqnarray}\label{107}\begin{split}
\|v_{i,t}-\bar{v}_{t}\|_\mathbb{F}&\leq\left\|\left([A^{t}]_{i\cdot}-\frac{{\textbf{1}}_n^T}{n}\right)\otimes I_{m_2}\right\|_\mathbb{F}\|v_0\|_\mathbb{F}\\
&~~~+\sum_{k=1}^t\left\|\left([A^{t}]_{i\cdot}-\frac{{\textbf{1}}_n^T}{n}\right)\otimes I_{m_2}\right\|_\mathbb{F}\|d_{k}\|_\mathbb{F}.
\end{split}\end{eqnarray}
Note that
\begin{equation*}\begin{split}
\|d_t\|_\mathbb{F}&\leq \sum_{i=1}^n\|\nabla_{22}g_i(x_{i,t},y_{i,t})-\nabla_{22}g_i(x_{i,t-1},y_{i,t})\|_\mathbb{F}\\
&~~~+\sum_{i=1}^n\|\nabla_{22}g_i(x_{i,t-1},y_{i,t})-\nabla_{22}g_i(x_{i,t-1},y_{i,t-1})\|_\mathbb{F}\\
&\leq \sum_{i=1}^nL\|x_{i,t}-x_{i,t-1}\|+\sum_{i=1}^nL\|{y}_{i,t}-{y}_{i,t-1}\|\\
&\leq 2nLK\eta_{t-1}+\sum_{i=1}^nL\|{y}_{i,t}-{y}_{i,t-1}\|
\end{split}\end{equation*}
where the last inequality results from (\ref{s4}) and the boundedness of $\Omega_i$ in Assumption 2.
Moreover, by
(\ref{104-22}), we have $\|\sigma(x_{t})-\sigma(x_{t-1})\|\leq\ell\|x_{t}-x_{t-1}\|$. Then, for any $\in\mathcal{V}$,
\begin{equation*}\begin{split}
&\|{y}_{i,t}-{y}_{i,t-1}\|\\&=\|({y}_{i,t}-\sigma(x_t))+(\sigma(x_t)-\sigma(x_{t-1}))\\
&~~~~+(\sigma(x_{t-1})-{y}_{i,t-1})\|\\
&\leq\|{y}_{i,t}-\sigma(x_t)\|+\|\sigma(x_t)-\sigma(x_{t-1})\|\\
&~~~+\|\sigma(x_{t-1})-{y}_{i,t-1}\|\\
&\leq \mathcal{O}\left(\gamma^t+\sum_{k=1}^t\gamma^{t-k}\eta_{k-1}+2K\sqrt{n}\ell\eta_{t-1}\right)
\end{split}\end{equation*}
 where the second inequality holds by using Lemma \ref{LE2}, (\ref{s4}) and the compactness of $\Omega_i$.
Then,
\begin{equation}\label{108}\begin{split}
\|d_t\|_\mathbb{F}&\leq 2nLK(1+\sqrt{n}\ell)\eta_{t-1}\\
&~~~+\frac{2n\rho L }{\gamma}\left(\gamma^t+\sum_{k=1}^t\gamma^{t-k}\eta_{k-1}\right).
\end{split}\end{equation}
By (\ref{107}) and (\ref{108}), using Lemma \ref{LE1} results in that
\begin{equation*}\begin{split}
&\|v_{i,t}-\bar{v}_{t}\|_\mathbb{F}\\
&\leq \frac{2n^2LK(1+\sqrt{n}\ell)\sqrt{{m_2}}}{\rho}\sum_{k=0}^t\rho^{t-k}\eta_{k}\\
& ~~~~+n\sqrt{{m_2}}\|v_0\|_\mathbb{F}\rho^t+\frac{2n^2L\rho\sqrt{{m_2}}}{\gamma\rho}\sum_{k=0}^t\rho^{t-k}\gamma^{k}\\ &~~~~+\frac{2n^2L\rho\sqrt{{m_2}}}{\gamma\rho}\sum_{k=0}^t\rho^{t-k}\sum_{s=0}^k\gamma^{k-s}\eta_{s}.
\end{split}\end{equation*}
Let $\tau=\max\left(\gamma, \rho\right)$, and recall the fact that $\bar{v}_t=\bar{G}_{t}$, it immediately leads to the validity of the result.\QEDA
\subsection{Proof of Lemma 6}
Denote $M_t=n\bar{G}_t$, since matrix $M_t$ is positive definite, then $z_{i, t}^*=-M_t^{-1}\nabla_2J_i(x_{i,t},y_{i,t})$ is the unique solution to equation $M_tz+\nabla_2J_i(x_{i,t},y_{i,t})=0$. Accordingly, $M_tz^*_t+\nabla_2J_i(x_{i,t}, y_{i,t})=0$, $\forall~t\geq0$ and $i\in\mathcal{V}$. By (\ref{s3}) in Algorithm 1, we have
\begin{equation}\label{6-0}\begin{split}
 &\|z_{i,t+1}-z_{i, t}^*\|^2\\
&=\|z_{i,t}-z_{i, t}^*\|^2-2\alpha\langle z_{i,t}-z_{i, t}^*, M_t(z_{i,t}-z_{i, t}^*)\rangle\\
&~~~+2\alpha\langle z_{i,t}-z_{i, t}^*, M_tz_{i,t}-nv_{i,t}z_{i,t}\rangle\\
&~~~+\alpha^2\|nv_{i,t}z_{i,t}-M_tz_{i, t}^*\|^2\\
&\leq(1-2n\mu\alpha)\|z_{i,t}-z_{i, t}^*\|^2\\
&~~~+2\alpha\langle z_{i,t}-z_{i, t}^*, M_tz_{i,t}-nv_{i,t}z_{i,t}\rangle\\
&~~~+\alpha^2\|nv_{i,t}z_{i,t}-M_tz_{i, t}^*\|^2\\
\end{split}\end{equation}
where the inequality holds true due to the strong monotonicity of $g_i(x_{i,t}, \cdot)$. It is not difficult to verify that $\|z_{i, t}^*\|\leq \frac{K_1}{n\mu}$, then there holds
\begin{equation}\label{eq109}\begin{split}
    &\langle z_{i,t}-z_{i, t}^*, M_tz_{i,t}-nv_{i,t}z_{i,t}\rangle\\
    &\leq\|z_{i,t}-z_{i, t}^*\|\| (M_t-nv_{i,t})z_{i,t}\|\\
    &\leq\frac{n\mu}{2}\|z_{i,t}-z_{i, t}^*\|^2+\frac{1}{2n\mu}\| (M_t-nv_{i,t})z_{i,t}\|^2\\
    &\leq\frac{n\mu}{2}\|z_{i,t}-z_{i, t}^*\|^2+\frac{1}{n\mu}\| M_t-nv_{i,t}\|_\mathbb{F}^2\|z_{i,t}-z_{i, t}^*\|^2\\
    &~~~+\frac{K_1^2}{n^3\mu^3}\| M_t-nv_{i,t}\|_\mathbb{F}^2\\
 \end{split}\end{equation}
where the second inequality results by using the Young's inequality theory, and the last one holds due to the fact that $\|u+v\|^2\leq 2\|u\|^2+2\|v\|^2$ for any vectors $u, v\in\mathbb{R}^{m_2}$.
Moreover,
\begin{equation}\label{eq110}\begin{split}
  &\|nv_{i,t}z_{i,t}-M_tz_{i, t}^*\|^2\\
&=\|(nv_{i,t}z_{i,t}-M_tz_{i,t})+(M_tz_{i,t}-M_tz_{i, t}^*)\|^2\\
  &\leq2\|nv_{i,t}-M_t\|_\mathbb{F}^2\|z_{i,t}\|^2+2\|M_t\|_\mathbb{F}^2\|z_{i,t}-z_{i, t}^*\|^2\\
  &\leq4\|nv_{i,t}-M_t\|_\mathbb{F}^2\|z_{i,t}-z_{i, t}^*\|^2+4\|z_{i, t}^*\|^2\| M_t-nv_{i,t}\|_\mathbb{F}^2\\
  &~~~+2\|M_t\|_\mathbb{F}^2\|z_{i,t}-z_{i, t}^*\|^2\\
  &\leq\left(4\|nv_{i,t}-M_t\|_\mathbb{F}^2+2n^2K_0^2\right)\|z_{i,t}-z_{i, t}^*\|^2\\
  &~~~+\frac{4K_1^2}{n^2\mu^2} \| M_t-nv_{i,t}\|_\mathbb{F}^2\\
  \end{split}\end{equation}
where the last inequality holds by using the boundedness of $\nabla_{22}g_i(x_{i,t},y_{i,t})$ in Assumption 2. By (\ref{6-0}), (\ref{eq109}) and (\ref{eq110}), and using the fact that $M_t=n\bar{G}_t$, we have
\begin{equation*}\begin{split}
  \|z_{i,t+1}-z_{i, t}^*\|^2&\leq\Big(\vartheta^2+r_1^2\| \bar{G}_t-v_{i,t}\|_\mathbb{F}^2\Big)\|z_{i,t}-z_{i, t}^*\|^2\\
  &~~~+r_2^2\| \bar{G}_t-v_{i,t}\|_\mathbb{F}^2.\\
  \end{split}\end{equation*}
Due to the fact that $0<\alpha<\frac{\mu}{2nK_0^2}$ and $\mu<K_0$, there holds that $0<1-{n\mu\alpha}+2n^2K_0^2\alpha^2<1$. Thus,
\begin{equation}\label{eq111}\begin{split}
  &\|z_{i,t+1}-h_{i, t+1}\|\\
  &\leq \|z_{i,t+1}-z_{i, t}^*\|+\|h_{i, t+1}-h_{i, t}\|+\|z_{i, t}^*-h_{i,t}\|\\
&\leq\Big(\vartheta+r_1\| \bar{G}_t-v_{i,t}\|_\mathbb{F}\Big)\|z_{i,t}-h_{i,t}\|\\
  &~~~+r_2\| \bar{G}_t-v_{i,t}\|_\mathbb{F}+\|h_{i,t+1}-h_{i, t}\|\\
 &~~~+\Big(1+\vartheta+r_1\| \bar{G}_t-v_{i,t}\|_\mathbb{F}\Big)\|z_{i,t}^*-h_{i,t}\|.\\
  \end{split}\end{equation}
Denote $H_t=\sum_{i=1}^n\nabla_{22}g_i(x_{i,t},\sigma(x_t))$, by the strong convexity of $\sum_{i=1}^ng_i(\cdot,\cdot)$ with respect to the second argument, one has $\|M_t^{-1}\|_\mathbb{F}\leq \frac{1}{\mu}$ and $\|H_t^{-1}\|_\mathbb{F}\leq \frac{1}{\mu}$. Furthermore,
\begin{equation}\label{6-2}\begin{split}
&\|z_{i, t}^*-h_{i,t}\|\\
&=\|M_t^{-1}\nabla_2J_i(x_{i,t},y_{i,t})-H_t^{-1}\nabla_2J_i(x_{i,t},y_{i,t})\\
&~~~~+H_t^{-1}\nabla_2J_i(x_{i,t},y_{i,t})-H_t^{-1}\nabla_2J_i(x_{i,t},\sigma(x_t))\|\\
&\leq \|\nabla_2J_i(x_{i,t},y_{i,t})\|\|M_t^{-1}-H_t^{-1}\|_\mathbb{F}\\
&~~~~+\|H_t^{-1}\|_\mathbb{F}\|\nabla_2J_i(x_{i,t},y_{i,t})-\nabla_2J_i(x_{i,t},\sigma(x_t))\|\\
&\leq K_1\|M_t^{-1}(H_t-M_t)H_t^{-1}\|+\frac{L}{\mu}\|y_{i,t}-\sigma(x_t)\|\\
&\leq r_3\sum_{i=1}^n\|y_{i,t}-\sigma(x_t)\|
\end{split}\end{equation}
where the second inequality is true due to the boundedness of $\nabla_2J_i$ in Assumption 2 and its Lipschitz
continuity in Assumption 3, and the last one holds by the Lipschitz
continuity of $\nabla_{22}g_i$. Similarly, we have
\begin{equation}\label{113}\begin{split}
&\|h_{i,t+1}-h_{i,t}\|\\
&\leq \|\nabla_2J_i(x_{i,t+1},\sigma(x_{t+1}))\|\|H_{t+1}^{-1}-H_t^{-1}\|_\mathbb{F}\\
&~~~~+\|H_t^{-1}\|_\mathbb{F}\|\nabla_2J_i(x_{i,t+1},\sigma(x_{t+1}))-\nabla_2J_i(x_{i,t},\sigma(x_t))\|\\
&\leq K_1\|H_{t+1}^{-1}(H_{t}-H_{t+1})H_t^{-1}\|+\frac{2L}{\mu}\|\sigma(x_{t+1})-\sigma(x_t)\|\\
&\leq \frac{nK_1L}{\mu^2}\sum_{i=1}^n\|y_{i,t}-\sigma(x_t)\|+\frac{2L\ell}{\mu}\|x_{t+1}-x_t\|.\\
\end{split}\end{equation}
By (\ref{eq111})-(\ref{113}), and using (\ref{s4}), we immediately achieve (\ref{eq116}). Thus, the validity of the results is verified.\QEDA
\subsection{Proof of Theorem 1}
Before providing the proof of Theorem 1, the following two lemmas are necessary.
\begin{lemma}\label{LE6}
Given a nonnegative sequence $\{\xi_t\}_{t=0,1,\cdots}$ such that
 \begin{equation*}\begin{split}
\xi_{t+1}\leq \Big(1-\frac{b_1}{t+a}\Big)\xi_t+\frac{b_2}{(t+a)^2}+b_3c^t+\frac{b_4}{t+a}
 \end{split}
 \end{equation*}
for some $b_i>0$, $i=1, \cdots, 4$, if $b_1\geq1$, $a>b_1$ and $0<c<1$, then for any $t\geq0$,
 \begin{equation}\label{LEQ16}\begin{split}
\xi_{t}&\leq  \frac{c_1(\ln(t+a)+1)+2^{b_1}}{t-1+a}+c_2c^t+2b_4
 \end{split}
 \end{equation}
 where $c_1=(a\xi_0+2b_2+2^{b_1}b_2+b_3)\frac{c^{1-\lceil 1+a\rceil}}{(1-c)^2}$, $c_2=\frac{2^{b_1}b_3}{c}$.
\end{lemma}
\begin{proof} It is not difficult to verify that (\ref{LEQ16}) holds if either $t=0$ or $t=1$. In what follows, we just verify the case when $t\geq2$.
Using the inequality $\exp(-r)\geq 1-r$ for any $r>0$ yields
\begin{equation}\label{2-0}\begin{split}
&\prod_{k=s}^t\Big(1-\frac{b_1}{k+a}\Big)\leq \exp\Big(-\sum_{k=s}^t\frac{r_1}{k+a}\Big)\\
&\leq \exp(b_1(\ln (s+a)-\ln (t+a)))\\
&=\Big(\frac{s+a}{t+a}\Big)^{b_1}
\end{split}\end{equation}
where the inequality holds due to $\sum_{k=s}^t\frac{1}{k+a}\geq \int_s^t\frac{1}{s+a}ds=\ln(t+a)-\ln(s+a)$.
By the iteration of $\xi_t$,  there holds that for any $t\geq2$,
\begin{equation*}\label{2-1}\begin{split}
\xi_t&\leq \prod_{k=0}^{t-1}\Big(1-\frac{b_1}{k+a}\Big)\xi_0+b_2\sum_{k=0}^{t-1}\prod_{s=k+1}^{t-1}\Big(1-\frac{b_1}{s+a}\Big)\frac{1}{(k+a)^2}\\
&~~~~+b_3\sum_{k=0}^{t-1}\prod_{s=k+1}^{t-1}\Big(1-\frac{b_1}{s+a}\Big)c^k\\
&~~~~+b_4\sum_{k=0}^{t-1}\prod_{s=k+1}^{t-1}\Big(1-\frac{b_1}{s+a}\Big)\frac{1}{k+a}\\
&\leq \Big(\frac{a}{t-1+a}\Big)^{b_1}\xi_0+b_2\sum_{k=0}^{t-1}\frac{(k+1+a)^{b_1}}{(t-1+a)^{b_1}(k+a)^2}\\
\end{split}\end{equation*}
\begin{equation}\label{2-1}\begin{split}
&~~~~+\sum_{k=0}^{t-1}\frac{b_3(k+1+a)^{b_1}c^k}{(t-1+a)^{r_1}}+\sum_{k=0}^{t-1}\frac{b_4(k+1+a)^{b_1}}{(t-1+a)^{b_1}(k+a)}\\
&\leq \frac{a\xi_0}{t-1+a}+\frac{2b_2}{t-1+a}\sum_{k=0}^{t-2}\frac{1}{k+a}+\frac{2^{b_1}b_2}{(t-1+a)^2}\\
&~~~~+\frac{b_3}{t-1+a}\sum_{k=0}^{t-2}(k+1+a)c^k+2^{b_1}b_3c^{t-1}\\
&~~~~+\frac{2^{b_1}b_4}{t-1+a}+{2b_4}\\
\end{split}\end{equation}
where the second inequality results from (\ref{2-0}), and the last one is true due to the facts that $r_1, a\geq 1$.
It is easy to verify that
\begin{equation}\label{2-2}\begin{split}
\sum_{k=0}^{t}\frac{1}{k+a}\leq \ln (t+a)+\frac{1}{a}.
\end{split}\end{equation}
Moreover, it is easy to compute that
\begin{equation}\label{2-3}\begin{split}
\sum_{k=0}^{t-2}c^k(k+1+a)&\leq c^{-\lceil 1+a\rceil}\sum_{k=0}^{t+\lceil 1+a\rceil}c^kk\\
&\leq \frac{c^{1-\lceil 1+a\rceil}}{(1-c)^2}.
\end{split}\end{equation}
By (\ref{2-1})-(\ref{2-3}), it immediately leads to the validity of Lemma \ref{LE6}.
\end{proof}
\begin{lemma}\label{LE7}
For any $a>0$ and $0<r<1$,
 \begin{equation*}\begin{split}
\sum_{s=0}^t\frac{r^{-s}}{s+a}\leq\frac{\delta_1r^{-t}}{t+a}+\frac{\delta_2 }{t+a}+\delta_3
 \end{split}
 \end{equation*}
 where $\delta_1=\frac{t_0+a}{r(\ln r^{-1}(t_0+a)-2)}$, $\delta_2=\frac{t_0r^{-t_0}(t_0+a)}{a}$, $\delta_3=\frac{2t_0r^{-t_0}(t_0+a)}{a^2(\ln r^{-1}(t_0+a)-2)}+\frac{1}{a}$ and $t_0=\max (\lceil \frac{2}{\ln r^{-1}}-a\rceil+1, 1)$.
\end{lemma}
\begin{proof}
It is noticed that
\begin{equation*}\begin{split}
\sum_{s=0}^t\frac{r^{-s}}{s+a}\leq \frac{1}{a}+\int_0^{t+1}\frac{r^{-s}}{s+a}ds.
 \end{split}
 \end{equation*}
For $t\geq t_0$, one has
\begin{equation*}\begin{split}
&\int_0^{t+1}\frac{r^{-s}}{s+a}ds\\
&=\frac{r^{-s}}{(s+a)\ln r^{-1}}\Big|_{s=0}^{t+1}-\int_0^{t+1}\frac{r^{-s}}{\ln r^{-1}}d\Big(\frac{1}{s+a}\Big)\\
 \end{split}\end{equation*}
\begin{equation}\label{8}\begin{split}
&\leq \frac{r^{-t}}{r\ln r^{-1}(t+a)}+\frac{2}{\ln r^{-1}}\int_0^{t_0}\frac{r^{-s}}{(s+a)^2}ds\\
&~~~~+\frac{2}{\ln r^{-1}}\int_{t_0}^{t+1}\frac{r^{-s}}{(s+a)^2}ds\\
&\leq \frac{r^{-t}}{r\ln r^{-1}(t+a)}+\frac{2t_0r^{-t_0}}{a^2\ln r^{-1}}\\
&~~~~+\frac{2}{\ln r^{-1}(t_0+a)}\int_{0}^{t+1}\frac{r^{-s}}{s+a}ds.
 \end{split}\end{equation}
Due to the fact that $\frac{2}{\ln r^{-1}(t_0+a)}\leq 1$, inequality (\ref{8}) implies that for any $t\geq t_0 $,
\begin{equation}\label{8-1}\begin{split}
\int_{0}^{t+1}\frac{r^{-s}}{s+a}ds&\leq \frac{t_0+a}{r(\ln r^{-1}(t_0+a)-2)}\frac{r^{-t}}{t+a}\\
&~~~~+\frac{2t_0r^{-t_0}(t_0+a)}{a^2(\ln r^{-1}(t_0+a)-2)}.
\end{split}\end{equation}
Moreover, for $t<t_0$, we have $\int_0^{t+1}\frac{r^{-s}}{s+a}ds\leq \int_0^{t_0}\frac{r^{-s}}{s+a}ds\leq \frac{t_0r^{-t_0}}{a}$. Together with (\ref{8-1}), it immediately leads to the validity of the result.
\end{proof}
 Then we can provide the proof of Theorem 1.

\textbf{The proof of Theorem 1.} By (\ref{s4}) in Algorithm 1 and using Lemma \ref{LE000}, we have
\begin{equation}\label{15}\begin{split}
&\sum_{i=1}^n\|x_{i,t+1}-x^*_{i}\|^2\\
&=\sum_{i=1}^n\Big\| (1-\eta_t)(x_{i,t}-x^*_{i})+\eta_tP_{\Omega_i}\big[x_{i,t}-k\hat{F}_{i,t}\big]\\
&~~~-\eta_tP_{\Omega_i}[x_{i}^*-k{\nabla _{{x_i}}}J_{i}(x_i^*,\sigma(x^*))]\Big\|^2\\
&\leq \sum_{i=1}^n\|x_{i,t}-x^*_i\|^2+2k^2\eta_t\sum_{i=1}^n\|\hat{F}_{i,t}-\nabla_{x_i}J_i(x_{i,t},\sigma(x_t))\|^2\\
&~~~-2k\eta_t\sum_{i=1}^n\Big\langle x_{i,t}-x^*_i, \hat{F}_{i,t}-\nabla_{x_i}J_i(x^*_i,\sigma(x^*))\Big\rangle\\
&~~~+2k^2\eta_t\|F(x_t)-F(x^*)\|^2\\
&\leq \sum_{i=1}^n\|x_{i,t}-x^*_i\|^2+2k^2\eta_t\sum_{i=1}^n\|\hat{F}_{i,t}-\nabla_{x_i}J_i(x_{i,t},\sigma(x_t))\|^2\\
&~~~-2k\eta_t\sum_{i=1}^n\Big\langle x_{i,t}-x^*_i, \hat{F}_{i,t}-\nabla_{x_i}J_i(x^*_i,\sigma(x^*))\Big\rangle\\
&~~~+2k^2L_0\eta_t\|x_t-x^*\|^2\\
\end{split}\end{equation}
where the first inequality holds by using the Jensen's inequality theory and the non-expansive property of projection, and the second one results from Lemma 2. Using the strong monotonicity of the pseudo-gradient mapping $F(x)$ yields that
\begin{equation}\label{15-2}\begin{split}
&-\sum_{i=1}^n\Big\langle x_{i,t}-x^*_i,\hat{F}_{i,t}-\nabla_{x_i}J_i(x^*_i,\sigma(x^*))\Big\rangle\\
&=-\sum_{i=1}^n\Big\langle x_{i,t}-x^*_i,\hat{F}_{i,t}-\nabla_{x_i}J_i(x_{i,t},\sigma(x_t))\Big\rangle\\
&~~~-\Big\langle x_{t}-x^*,F(x_t)-F(x^*)\Big\rangle\\
&\leq 2K\sum_{i=1}^n\|\hat{F}_{i,t}-\nabla_{x_i}J_i(x_{i,t},\sigma(x_t))\|-\theta\|x_t-x^*\|^2.
\end{split}\end{equation}
Moreover, by the fact that $\|h_{i,t}\|\leq\frac{K_1}{\mu}$, we have
\begin{equation}\label{15-3}\begin{split}
&\|\hat{F}_{i,t}-\nabla_{x_i}J_i(x_{i,t},\sigma(x_t))\|\\
& \leq \|\nabla_1J_i(x_{i,t},y_{i,t})-\nabla_1J_i(x_{i,t},\sigma(x_t))\|\\
&~~~~+\|\nabla_{21}g_i(x_{i,t},y_{i,t})z_{i,t}-(\nabla_{x_i}\sigma(x_t))^T\nabla_2J_i(x_{i,t},\sigma(x_t))\|\\
&\leq L\|y_{i,t}-\sigma(x_t)\|+\|\nabla_{21}g_i(x_{i,t},y_{i,t})z_{i,t}\\
&~~~~-\nabla_{21}g_i(x_{i,t},\sigma(x_t))h_{i,t}\|\\
&\leq (1+K_1/\mu)L\|y_{i,t}-\sigma(x_t)\|+K_3\|z_{i,t}-h_{i,t}\|
\end{split}\end{equation}
where the second inequality results from the  Lipschitz continuity of $\nabla_1J_i(\cdot, \cdot)$ with respect to the second argument, and the last one holds by using the Lipschitz continuity of $\nabla_{21}g_i(\cdot, \cdot)$ with respect to the second argument. By (\ref{15})-(\ref{15-3}), there holds that
\begin{equation}\label{15-4}\begin{split}
&\|x_{t+1}-x^*\|^2\\
&\leq(1-2k(\theta-kL_0)\eta_t) \|x_{t}-x^*\|^2+\\
&~~~4k^2(1+K_1/\mu)^2L^2\eta_t\sum_{i=1}^n\|y_{i,t}-\sigma(x_t)\|^2\\
&~~~+4k^2K_3^2\eta_t\sum_{i=1}^n\|z_{i,t}-h_{i,t}\|^2\\
&~~~+4kK(1+K_1/\mu)L\eta_t\sum_{i=1}^n\|y_{i,t}-\sigma(x_t)\|\\
&~~~+4k KK_3\eta_t\sum_{i=1}^n\|z_{i,t}-h_{i,t}\|.\\
\end{split}\end{equation}
Using Lemma \ref{LE2} and Lemma \ref{LE7}, we have
\begin{equation}\label{eq1101}\begin{split}
&\|y_{i,t}-\sigma(x_t)\|\leq \mathcal{O}\left(\gamma^t+\frac{1}{t+a} \right).\\
\end{split}\end{equation}
Note that $t\tau^{t}\leq\mathcal{O}\left(\frac{1}{t+a}\right)$. Combining Lemma \ref{LE4} and Lemma \ref{LE7} results in that
\begin{equation}\label{eq1102}\begin{split}
\|v_{i,t}-\bar{G}_{t}\|_\mathbb{F}&\leq \mathcal{O}\left(\frac{1}{t+a}+\tau^{t}\right).
\end{split}\end{equation}
Moreover, based on (\ref{eq116}) in Lemma \ref{LE5}, there holds
\begin{equation*}\begin{split}
  \|z_{i,t+1}-h_{i, t+1}\|&\leq\left(\vartheta+ \mathcal{O}\left(\frac{1}{t+a}+\tau^{t}\right)\right)\|z_{i,t}-h_{i, t}\|\\
 &~~~+ \mathcal{O}\left(\frac{1}{t+a}+\tau^{t}\right).
  \end{split}\end{equation*}
Note that $\frac{1}{t+a}+\tau^{t}$ decays, there must exist some finite $t_0>0$ such that $\mathcal{O}\left(\frac{1}{t+a}+\tau^{t}\right)\leq \frac{1-\vartheta}{2}$ for any $t\geq t_0$. Denote $\varepsilon=\frac{1+\vartheta}{2}$, by the fact that $0<\vartheta<1$, we have $0<\varepsilon<1$. Accordingly,
\begin{equation}\label{eq1103}\begin{split}
&\|z_{i,t+1}-h_{i, t+1}\|\leq\mathcal{O}\left(\frac{1}{t+a}+\varepsilon^{t}\right).
  \end{split}\end{equation}
By (\ref{15-4}), (\ref{eq1101}) and (\ref{eq1103}), we have
\begin{equation}\label{15-5}\begin{split}
\|x_{t+1}-x^*\|^2\leq&\left(1-\frac{2kb(\theta-kL_0)}{t+a}\right) \|x_{t}-x^*\|^2\\
&+\mathcal{O}\left(\frac{1}{(t+a)^2}+\epsilon^{t}\right).
\end{split}\end{equation}
where $\epsilon=\max(\gamma,\tau, \varepsilon)$.  By the facts that $a> \max (2k(\theta-kL_0)b, b)$ and $b\geq \max \left(\frac{1}{2k(\theta-kL_0)}, 1\right)$, one knows that $2kb(\theta-kL_0)\geq1$, $a> 1$, and ${\frac{b}{t+a}}\leq1$. Then
using Lemma \ref{LE6} and letting $r_4=0$ immediately imply $\|x_t-x^*\|\leq\mathcal{O}\left(\sqrt{\ln t/{t}}+\sqrt{\epsilon}^{t}\right)$.  Together with the fact that $\mathcal{O}\left(\sqrt{\epsilon}^{t}\right)\leq\mathcal{O}\left(\sqrt{\ln t/{t}}\right)$, we have $\|x_t-x^*\|\leq\mathcal{O}\left(\sqrt{\ln t/{t}}\right)$. Furthermore, from (\ref{eq1101}), it follows that
\begin{equation}\label{15-115}\begin{split}
\|y_{i,t}-\sigma(x^*)\|&=\|(y_{i,t}-\sigma(x_t))+(\sigma(x_t)-\sigma(x^*))\|\\
&\leq\|y_{i,t}-\sigma(x_t)\|+\ell\|x_t-x^*\|\\
&\leq \mathcal{O}\left(\sqrt{\ln t/{t}}+\tau^t+\frac{1}{t+a} \right)\\
&= \mathcal{O}\left(\sqrt{\ln t/{t}} \right)
\end{split}\end{equation}
where the second inequality results from (\ref{104-22}). This completes the proof.\QEDA
\subsection{Proof of Lemma 7}
By (\ref{eq2001}), we have
\begin{equation}\label{eq2003}\begin{split}
&\delta_i\nabla_2 J_i(x_{i, t},y_{i}(\delta_i,x_t))+\sum_{j=1}^n\nabla_2 g_j(x_{j, t}, y_{i}(\delta_i,x_t))\\
&=0.
\end{split}\end{equation}
Taking the derivative results in that
\begin{equation*}\begin{split}
&\left(\sum_{j=1}^n\nabla_{22} g_j(x_{j, t}, y_i(s,x_t))\right)\frac{dy_{i}(s,x_t)}{ds}+\nabla_2 J_i(x_{i, t},\\
&y_{i}(s,x_t))+\frac{sdy_{i}(s,x_t)}{ds}\nabla_{22} J_i(x_{i, t},y_{i}(s,x_t))=0, s\in[0, \delta_i].
\end{split}\end{equation*}
Thus,
\begin{equation}\label{eq2005}\begin{split}
&\frac{dy_{i}(s,x_t)}{ds}=-\Big(\sum_{j=1}^n\nabla_{22} g_j(x_{j, t}, y_{i}(s,x_t))+s\nabla_{22} J_i(x_{i, t},\\
&~~~y_{i}(s,x_t))\Big)^{-1}\nabla_2 J_i(x_{i, t},y_{i}(s,x_t)),~~s\in[0, \delta_i].\\
\end{split}\end{equation}
Letting $s=0$, by (\ref{a2}) and (\ref{eq2001}), we have $y_{i}(0,x_t)=\sigma(x_t)$ for any $i\in \mathcal{V}$. Then there holds that
\begin{equation}\label{eq2006}\begin{split}
&~~~~~~\frac{dy_{i}(0,x_t)}{ds}=\frac{dy_{i}(s,x_t)}{ds}\Big|_{s=0}=\\
&-\Big(\sum_{j=1}^n\nabla_{22} g_j(x_{j, t}, \sigma(x_t))\Big)^{-1}\nabla_2 J_i(x_{i, t},\sigma(x_t)).\\
\end{split}\end{equation}
By (\ref{eq2006}), using Proposition 1, we have
\begin{equation}\label{eq2007}\begin{split}
&(\nabla_{21}g_i(x_{i, t}, \sigma(x_t)))^T\frac{dy_{i}(0,x_t)}{ds}\\
&=(\nabla_{x_i} \sigma(x_t))^T\nabla_2 J_i(x_{i, t},\sigma(x_t)).\\
\end{split}\end{equation}
From (\ref{eq2007}), by the fact that $(\nabla_{21}g_i(x_{i, t}, \sigma(x_t)))^T=\nabla_{12}g_i(x_{i, t}, \sigma(x_t))$, and using the fundamental theorem of calculus, it follows that
\begin{equation}\label{eq2008}\begin{split}
&\|D_{i,t}-(\nabla_{x_i} \sigma(x_t))^T\nabla_2 J_i(x_{i, t},\sigma(x_t))\Big\|\\
&=\frac{1}{\delta_i}\Big\|\int_0^{\delta_i}\nabla_{12} g_i(x_{i,t}, y_{i}(s,x_t))\frac{dy_{i}(s,x_t)}{ds}ds\\
&~~~-\nabla_{12}g_i(x_{i, t}, \sigma(x_t))\frac{dy_{i}(0,x_t)}{ds}\Big\|\\
&=\frac{1}{\delta_i}\Big\|\int_0^{\delta_i}\Big(\nabla_{12} g_i(x_{i,t}, y_{i}(s,x_t))\frac{dy_{i}(s,x_t)}{ds}\\
&~~~-\nabla_{12}g_i(x_{i, t}, y_{i}(0,x_t))\frac{dy_{i}(s,x_t)}{ds}\Big)ds\\
&~~~+\frac{1}{\delta_i}\int_0^{\delta_i}\Big(\nabla_{12} g_i(x_{i,t}, \sigma(x_t))\frac{dy_{i}(s,x_t)}{ds}\\
&~~~-\nabla_{12}g_i(x_{i, t}, \sigma(x_t))\frac{dy_{i}(0,x_t)}{ds}\Big)ds\Big\|\\
&\leq\frac{L}{\delta_i}\int_0^{\delta_i}\|y_{i}(s,x_t)-y_{i}(0,x_t)\|\left\|\frac{dy_{i}(s,x_t)}{ds}\right\|ds\\
&~~~+\frac{K_2}{\delta_i}\int_0^{\delta_i}\left\|\frac{dy_{i}(s,x_t)}{ds}-\frac{dy_{i}(0,x_t)}{ds}\right\|ds.\\
\end{split}\end{equation}
By (\ref{eq2005}), and using Assumption \ref{as2} yields that $\|\frac{dy_{i}(s,x_t)}{s}\|\leq\frac{K_1}{\sqrt{n}\mu}$. As a result, $\|y_{i}(s_1, x_t)-y_{i}(s_2, x_t)\|\leq \frac{K_1|s_1-s_2|}{\sqrt{n}\mu}$ for any $s_1, s_2\in[0, \delta_i]$. Employing the same approaches used in (\ref{6-2}), we can achieve that
\begin{equation}\label{eq2009}\begin{split}
&\left\|\frac{dy_{i}(s,x_t)}{ds}-\frac{dy_{i}(0,x_t)}{ds}\right\|\\
&\leq \frac{nK_1L}{\mu^2}(\|y_{i}(s,x_t)-y_{i}(0,x_t)\|+\hbar s)\\
&~~~+\frac{nL}{\mu}\|y_{i,t}(s,x_t)-y_{i}(0,x_t)\|\\
&\leq \mathcal{C}_0s.
\end{split}\end{equation}
It follows from (\ref{eq2008}) and (\ref{eq2009}) that
\begin{equation}\label{eq2010}\begin{split}
&\|D_{i,t}-(\nabla_{x_i} \sigma(x_t))^T\nabla_2 J_i(x_{i, t},\sigma(x_t))\|\\
&\leq\frac{LK_1^2}{n\delta_i\mu^2}\int_0^{\delta_i}s ds+\frac{\mathcal{C}_0K_2}{\delta_i}\int_0^{\delta_i}s ds\\
&=  \mathcal{C}\delta_i.
\end{split}\end{equation}
Thus, the validity of the result is verified.\QEDA
\subsection{Proof of Theorem 2}
To prove Theorem 2, we present the error between the estimate $w_{i,t}$ and the real value of $y_{i}(\delta_i, x_t)$.
\begin{lemma}\label{LE12}
Under Assumptions 1-3 and $5$, by Algorithm 2, for any $t\geq0$ and $i, k\in\mathcal{V}$,
\begin{equation}\label{eq4001}\begin{split}
\|w_{ik,t}-y_{i}(\delta_i, x_t)\|\leq\mathcal{O}\left(\gamma^t+\sum_{s=1}^t\gamma^{t-s}\eta_{s-1}\right). \\
\end{split}\end{equation}
\end{lemma}
\begin{proof} Based on  Assumptions 3 and 5, it is easy to conclude that $\nabla_2g_k(x_k,\cdot)-c_{ki}\delta_i\nabla_2J_i(x_{i},\cdot)$ is $\mu-$strongly monotone and $L_1-$Lipschitz continuous, where $L_1=L+\delta_i\hbar$. Moreover, by (\ref{eq2003}), letting $x_t=x$, and taking the Jacobian matrix with respect to $x_i$, we have
\begin{equation*}
\begin{split}
&\delta_i\nabla_{21} J_i(x_{i},y_{i}(\delta_i,x)) + \nabla_{21}g_i(x_i, y_{i}(\delta_i,x))+ \nabla_{x_i} y_{i}(\delta_i, \\
&x)\Big(\sum_{i=1}^n \nabla_{22}g_i(x_i, y_{i}(\delta_i,x))+\delta_i\nabla_{22} J_i(x_{i},y_{i}(\delta_i,x))\Big)=0.
\end{split}
\end{equation*}
It follows that
\begin{equation}\label{3104}
\begin{split}
\nabla_{x_i} y_{i}(\delta_i,x)=&-\Big(\delta_i\nabla_{21} J_i(x_{i},y_{i}(\delta_i,x))+ \nabla_{21}g_i(x_i, \\
& y_{i}(\delta_i,x))\Big)\Big(\sum_{i=1}^n \nabla_{22}g_i(x_i, y_{i}(\delta_i,x))\\
&+\delta_i\nabla_{22} J_i(x_{i},y_{i}(\delta_i,x))\Big)^{-1}.
\end{split}
\end{equation}
By (\ref{3104}), using Assumptions 2 and 5, we know that $\nabla_{x_i} y_{i}(\delta_i,x)$ is bounded for any $i\in\mathcal{V}$. Thus, $y_{i}(\delta_i,x)$ is Lipschitz continuous.  Then using the same approaches used in the proof of Lemma \ref{LE2}, the validity of (\ref{eq4001}) is verified.
\end{proof}

\textbf{Proof of Theorem 2.} By (\ref{15}) and (\ref{15-2}), we have
\begin{equation}\label{3105}\begin{split}
&\sum_{i=1}^n\|x_{i,t+1}-x^*_{i}\|^2\\
&\leq (1-2k(\theta-kL_0)\eta_t)\sum_{i=1}^n\|x_{i,t}-x^*_i\|^2\\
&~~~+2k^2\eta_t\sum_{i=1}^n\|\hat{F}_{i,t}-\nabla_{x_i}J_i(x_{i,t},\sigma(x_t))\|^2\\
&~~~-4kK\eta_t\sum_{i=1}^n\|\hat{F}_{i,t}-\nabla_{x_i}J_i(x_{i,t},\sigma(x_t))\|.\\
\end{split}\end{equation}
It is noticed that
\begin{equation}\label{3109}\begin{split}
&\|\hat{F}_{i,t}-\nabla_{x_i}J_i(x_{i,t},\sigma(x_t))\|\\
& \leq \|\nabla_1J_i(x_{i,t},y_{i,t})-\nabla_1J_i(x_{i,t},\sigma(x_t))\|\\
&~~~~+\|d_{i,t}-(\nabla_{x_i}\sigma(x_t))^T\nabla_2J_i(x_{i,t},\sigma(x_t))\|\\
& \leq \|\nabla_1J_i(x_{i,t},y_{i,t})-\nabla_1J_i(x_{i,t},\sigma(x_t))\|\\
&~~~~+\|D_{i,t}-(\nabla_{x_i}\sigma(x_t))^T\nabla_2J_i(x_{i,t},\sigma(x_t))\|\\
&~~~~+\|D_{i,t}-d_{i,t}\|\\
& \leq L\|y_{i,t}-\sigma(x_t)\|+\|D_{i,t}-d_{i,t}\|+\mathcal{C}\delta_i\\
\end{split}\end{equation}
where the last inequality holds by using the Lipschitz continuity in Assumption 3 and Lemma \ref{LE12}. Moreover, combining (\ref{s33}) and (\ref{eq113001}) results in that
\begin{equation}\label{3110}\begin{split}
\|D_{i,t}-d_{i,t}\|&\leq\frac{1}{\delta_i}\|\nabla_2 g_i(x_{i,t}, w_{ii,t})-\nabla_2 g_i(x_{i,t}, y_{i}(\delta_i,x_t))\|\\
&~~~+\frac{1}{\delta_i}\|\nabla_2 g_i(x_{i,t}, y_{i,t})-\nabla_2 g_i(x_{i,t}, \sigma(x_t))\|\\
&\leq\frac{L}{\delta_i}(\|w_{ii,t}-y_{i}(\delta_i,x_t)\|+\|y_{i,t}- \sigma(x_t)\|).
\end{split}\end{equation}
By (\ref{3105})-(\ref{3110}), we have
\begin{equation*}\begin{split}
&\sum_{i=1}^n\|x_{i,t+1}-x^*_{i}\|^2\\
&\leq (1-2k(\theta-kL_0)\eta_t)\sum_{i=1}^n\|x_{i,t}-x^*_i\|^2\\
&~~~+\frac{6k^2L^2\eta_t}{\delta_i^2}\sum_{i=1}^n  \|w_{ii,t}-y_{i}(\delta_i,x_t)\|^2 \\
&~~~+6k^2L^2\left(1+\frac{1}{\delta_i}\right)^2\eta_t\sum_{i=1}^n  \|y_{i,t}- \sigma(x_t)\|^2 +      \\
&~~~+4kK L\left(1+\frac{1}{\delta_i}\right)\eta_t\sum_{i=1}^n\|y_{i,t}- \sigma(x_t)\|+\\
&~~~\frac{4kKL\eta_t}{\delta_i}\sum_{i=1}^n\|w_{ii,t}-y_{i}(\delta_i,x_t)\|+\\
&~~~\left(4kK\mathcal{C}\sum_{i=1}^n\delta_i+6k^2\mathcal{C}^2\sum_{i=1}^n\delta_i^2\right)\eta_t.\\
\end{split}\end{equation*}
By the same approaches used in (\ref{15-4})-(\ref{15-5}), there holds that
\begin{equation*}\begin{split}
\|x_{t}-x^*\|^2\leq&\left(1-\frac{2kb(\theta-kL_0)}{t+a}\right) \|x_{t}-x^*\|^2\\
&+\mathcal{O}\left(\frac{1}{(t+a)^2}+\epsilon^{t}+\frac{\sum_{i=1}^n(\mathcal{C}\delta_i+\mathcal{C}^2\delta_i^2)}{t+a}\right)
\end{split}\end{equation*}
where $\epsilon$ is defined in (\ref{15-5}). Using Lemma \ref{LE6} again immediately implies (\ref{3120}), and inequality (\ref{3121}) can be achieved by using the same approach used in (\ref{15-115}). This completes the proof.\QEDA



%



\ifCLASSOPTIONcaptionsoff
  \newpage
\fi




\begin{thebibliography}{1}
\bibitem{Jensen}M. K. Jensen. Aggregative games. Handbook of Game Theory and Industrial Organization, volume I, pp. 66-92, 2018.
\bibitem{Alpcan} T. Alpcan, T. Basar. A game-theoretic framework for congestion control in general topology networks. Proceedings of the 41st IEEE Conference on Decision and Control, vol. 2: pp. 1218-1224, 2002.
\bibitem{Shahidehpour}H. Wu, M. Shahidehpour, A. Alabdulwahab, A. Abusorrah. A game theoretic approach to risk-based optimal bidding strategies for electric vehicle aggregators in electricity markets with variable wind energy resources. IEEE Transactions on Sustainable Energy, vol. 7, no. 1, pp. 374-385, 2015.
\bibitem{Algazin}G. I. Algazin, D. G. Algazina. Aggregate estimates of reflexive collective behavior
dynamics in a Cournot oligopoly model. Automation and Remote Control, vol. 85, no. 9, pp. 923-933, 2024.
 \bibitem{Y. You}Y. You, Q. Xu, C. Fischione. Hierarchical online game-theoretic framework for real-time energy trading in smart grid. IEEE Transactions on Smart Grid, vol. 15, no. 2, pp. 1634-1645, 2024.
\bibitem{Pave2}F. Salehisadaghiani, L. Pavel. Distributed Nash equilibrium seeking: A
gossip-based algorithm. Automatica, vol. 72, pp. 209-216, 2016.
\bibitem{Ye1} M. Ye, G. Hu. Game design and analysis for price-based demand response: An aggregate game approach. IEEE transactions on cybernetics, vol. 47, no. 3, pp.720-730, 2016.
\bibitem{Lu}K. Lu, G. Jing, L. Wang. Distributed algorithms for searching generalized Nash equilibrium of noncooperative games. IEEE Transactions on Cybernetics, vol. 49, no. 6, pp. 2362-2371, 2019.
\bibitem{Parise2020}F. Parise, S. Grammatico, B. Gentile, J. Lygeros. Distributed convergence to Nash equilibria in network and average aggregative games. Automatica, vol. 117, pp. 108959, 2020.

\bibitem{Paccagnan2019}D. Paccagnan, B. Gentile, F. Parise, M. Kamgarpour, J. Lygeros. Nash and Wardrop equilibria in aggregative games with coupling constraints. IEEE Transactions on Automatic Control, vol. 64, no. 4, pp. 1373-1388, 2019.
\bibitem{Zhu2021}Y. Zhu, W. Yu, G. Wen, G. Chen. Distributed Nash equilibrium seeking in an aggregative game on a directed graph. IEEE Transactions on Automatic Control, vol. 66, no. 6, pp. 2746-2753, 2021.
\bibitem{Lei2020}J. Lei, U. V. Shanbhag, J. Chen. Distributed computation of Nash equilibria for monotone aggregative games via iterative regularization. IEEE Conference on Decision and Control (CDC), pp. 2285-2290, 2020.



\bibitem{Ye2022}M. Ye, G. Hu, L. Xie, S. Xu. Differentially private distributed Nash equilibrium seeking for aggregative games. IEEE Transactions on Automatic Control, vol. 67, no. 5, pp. 2451-2458, 2022.


\bibitem{Wang2024}Y. Wang, A. Nedi\'{c}. Differentially private distributed algorithms for aggregative games with guaranteed convergence. IEEE Transactions on Automatic Control, vol. 69, no. 8, pp. 5168-5183, 2024.
    \bibitem{Yi2023}T. Wang, P. Yi, J. Chen, Distributed mirror descent method with operator extrapolation for stochastic aggregative games, Automatica, vol. 159, pp. 111356, 2024.
\bibitem{Belgioioso2021}G. Belgioioso, A. Nedi\'{c}, S. Grammatico. Distributed generalized Nash equilibrium seeking in aggregative games on time-varying networks. IEEE Transactions on Automatic Control, vol. 66, no. 5, pp. 2061-2075, 2021.

\bibitem{Carnevale2024}G. Carnevale, F. Fabiani, F. Fele, K. Margellos, G. Notarstefano. Tracking-dased distributed equilibrium seeking for aggregative games. IEEE Transactions on Automatic Control, vol. 69, no. 9, pp. 6026-6041, 2024.
\bibitem{Liang2017}S. Liang, P. Yi, Y. Hong. Distributed Nash equilibrium seeking for aggregative games with coupled constraints. Automatica, vol. 85, pp. 179-185, 2017.
\bibitem{Chen} M. Maljkovic, G. Nilsson, N. Geroliminis. On finding the leader's strategy in quadratic aggregative stackelberg pricing games. European Control Conference (ECC), IEEE, pp. 1-6, 2023.
\bibitem{Maljkovic} R. Li, G. Chen, D. Gan, H. Gu, J. L\"{u}. Stackelberg and Nash equilibrium computation in non-convex leader-follower network aggregative games. IEEE Transactions on Circuits and Systems I: Regular Papers, vol. 71, no. 2, pp. 898-909, 2024.
\bibitem{Fabiani} F. Fabiani, M. A. Tajeddini, H. Kebriaei, S. Grammatico. Local Stackelberg equilibrium seeking in generalized aggregative games. IEEE Transactions on Automatic Control, vol. 67, no. 2, pp. 965-970, 2022.
\bibitem{Shokri} M. Shokri, H. Kebriaei. Leader-follower network aggregative game with stochastic agents' communication and activeness. IEEE Transactions on Automatic Control, vol. 65, no. 12, pp. 5496-5502, 2021.
\bibitem{Facchinei}F. Facchinei, J. S. Pang, Finite-dimensional variational inequalities
and complementarity problems. Springer Science \& Business Media, 2007.


\bibitem{Lu1}K. Lu, G. Li, L. Wang. Online distributed algorithms for seeking generalized Nash equilibria in dynamic environments. IEEE Transactions on Automatic Control, vol. 66, no. 5, pp. 2289-2296, 2021.
\bibitem{Lu10}K. Lu. Online distributed algorithms for online noncooperative games with stochastic cost functions: High probability bound of regrets. IEEE Transactions on Automatic Control, vol. 69, no. 12, pp. 8860-8867, 2024.

\bibitem{Feijer}D. Feijer, F. Paganini. Stability of primal-dual gradient dynamics
and applications to network optimization. Automatica, vol. 46, no. 12,
pp. 1974-1981, 2010.
\bibitem{Ling}Q. Ling, W. Shi, G. Wu, A. Ribeiro. DLM: Decentralized linearized
alternating direction method of multipliers. IEEE Transactions on
Signal Processing, vol. 63, pp. 4051-4064, 2015.
\bibitem{Hamedani} E. Y. Hamedani, N. S. Aybat. A primal-dual algorithm with line search for general convex-concave saddle point problems. SIAM Journal on Optimization, vol. 31, no. 2, pp. 1299-1329, 2021.
\bibitem{Olfati} R. Olfati-Saber, J. A. Fax, R. M. Murray. Consensus and cooperation in networked multi-agent systems. Proceedings of the IEEE, vol. 95, no. 1, pp. 215-233, 2007.
\bibitem{Ren1}F. Xiao, L. Wang. Asynchronous consensus in continuous-time multi-agent systems with switching topology and time-varying delays. IEEE Transactions on Automatic Control, vol. 53, no. 8, pp. 1804-1816, 2008.
\bibitem{Ren}L. Wang, F. Xiao. Finite-time consensus problems for networks of dynamic agents. IEEE Transactions on Automatic Control, vol. 55, no. 4, pp. 950-955, 2010.
 \bibitem{Barshooi} A. Amirkhani, A. H. Barshooi. Consensus in multi-agent systems: a review. Artificial Intelligence Review, vol. 55, no. 5, pp. 3897-3935, 2022.


\bibitem{BU1}K. Scaman, F. Bach, S. Bubeck, Y. T. Lee, L. Massouli\'{e}. Optimal algorithms for smooth and strongly convex distributed optimization in networks. International Conference on Machine Learning, pp. 3027-3036, 2017.

\bibitem{BU2} A. Nedi\'{c}. Distributed Gradient Methods for Convex Machine Learning Problems in Networks: Distributed Optimization. IEEE Signal Processing Magazine, vol. 37, no. 3, pp. 92-101, 2020.

\bibitem{BU3}D. Kovalev, A. Salim, P. Richt\'{a}rik. Optimal and practical algorithms for smooth and strongly convex decentralized optimization. Advances in Neural Information Processing Systems, vol. 33, pp. 18342-18352, 2020.

\bibitem{Qu}S. Pu, A. Nedi\'{c}. Distributed stochastic gradient tracking methods. Mathematical Programming, vol. 187, no. 1, pp. 409-457, 2021.
\bibitem{Pu}G. Qu, N. Li. Harnessing smoothness to accelerate distributed optimization. IEEE Transactions on Control of Network Systems, vol. 5, no. 3,
pp. 1245-1260, 2018.


\bibitem{Duchi}J. C. Duchi, A. Agarwal, M. J. Wainwright. Dual averaging for
distributed optimization: Convergence analysis and network scaling.
IEEE Transactions on Automatic control, vol. 57, no. 3, pp. 592-606, 2012.
\bibitem{Olshevsky}   Nedi\'{c} A, Olshevsky A. Distributed optimization over time-varying directed graphs. IEEE Transactions on Automatic Control, vol. 60, no. 3, pp. 601-615, 2015.


\bibitem{Yang} Y. F. Yang, D. H. Li, L. Qi. A feasible sequential linear equation method for inequality constrained optimization. SIAM Journal on Optimization, vol. 13, no. 4, pp. 1222-1244, 2003.
\bibitem{Hwang}I. Hwang, B. Song, S. S. Soliman. A holistic view on hyper-dense heterogeneous and small cell networks. IEEE Communications Magazine, vol. 51, no. 6, pp. 20-27, 2013.
\bibitem{1AA} K. Ji, J. D. Lee, Y. Liang, et al. Convergence of meta-learning with task-specific adaptation over partial parameters. Advances in Neural Information Processing Systems, vol. 33, pp. 11490-11500, 2020.

\end{thebibliography}
%

\begin{IEEEbiography}[{\includegraphics[width=1in,height=1.25in,clip,keepaspectratio]{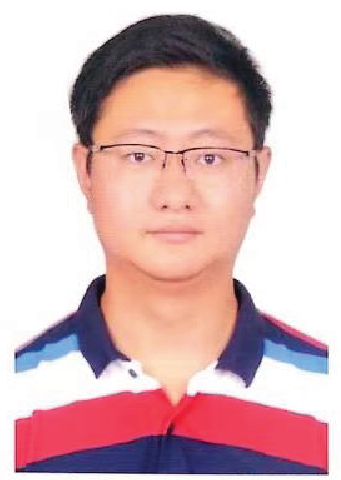}}]
{Kaihong Lu} received the Ph.D. degree in control theory and control engineering from Xidian University, Xi'an, China, in 2019. He is currently a Professor with the School of Electrical and Automation Engineering, Shandong University of Science and Technology, Qingdao, China. His current research interests include distributed optimization, game theory, and networked systems. Email: khong\_lu@163.com
\end{IEEEbiography}

\begin{IEEEbiography}[{\includegraphics[width=1in,height=1.25in,clip,keepaspectratio]{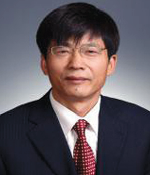}}]
{Huanshui Zhang} received the Ph.D. degree in control theory from Northeastern University, Shenyang, China, in 1997. He was a Postdoctoral Fellow at Nanyang Technological University, Singapore, from 1998 to 2001 and Research Fellow at Hong Kong Polytechnic University, Hong Kong, China, from 2001 to 2003. He currently holds a Professorship at Shandong University of Science and Technology, Qingdao, China. He was a Professor with the Harbin Institute of Technology, Shenzhen, China, from 2003 to 2006 and a Professor with Shandong university, Jinan, China, from 2006 to 2019. His interests include optimal estimation and control, time-delay systems, stochastic systems, signal processing and wireless sensor networked systems. Email: hszhang@sdu.edu.cn
\end{IEEEbiography}

\begin{IEEEbiography}[{\includegraphics[width=1in,height=1.25in,clip,keepaspectratio]{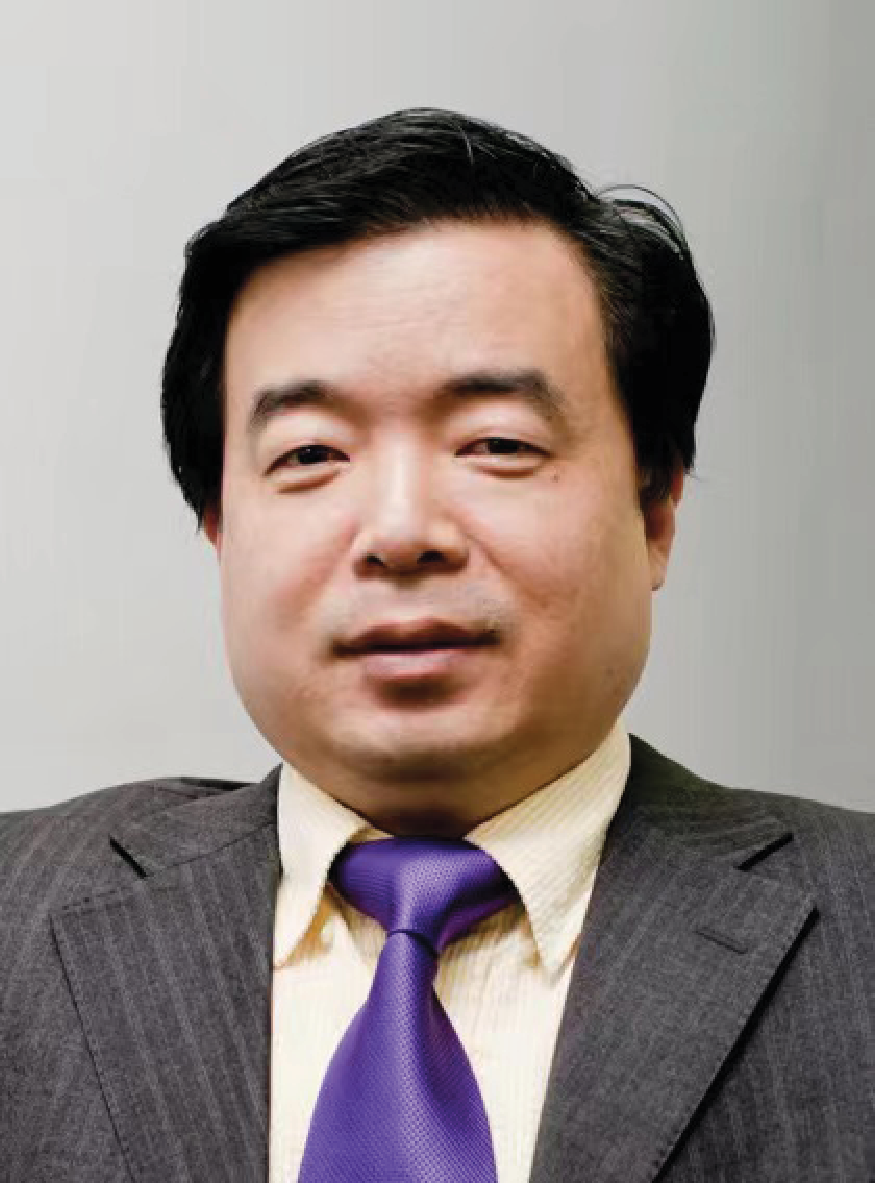}}]
{Long Wang} received the bachelor's degree in automatic control from Tsinghua University in 1986, and Ph.D. degree in Dynamics and Control from Peking University in 1992. During 1993-1997, He held postdoctoral research positions at the University of Toronto, Canada (with Professor Bruce A. Francis) and the German Aerospace Center, Munich, Germany (with Professor Juergen E. Ackermann). He is currently the Cheung-Kong Chair Professor of Dynamics and Control, and the Director of Center for Systems and Control of Peking University. His research interests include complex networked systems, evolutionary game dynamics, artificial intelligence, and bio-mimetic robotics. Email: longwang@pku.edu.cn
\end{IEEEbiography}

\end{document}